\documentclass[journal]{IEEEtran}
\usepackage{amsfonts} 
\ifCLASSINFOpdf \else \fi

\usepackage{amsmath}
\usepackage{amssymb} 
\usepackage[dvips]{graphicx} 
\usepackage{amsmath,amsfonts,amssymb} 
\usepackage{dsfont}
\usepackage{allfonts}
\usepackage{textcomp} 
\usepackage{times}
\usepackage{authblk}
\usepackage{cite}
\usepackage{graphicx}

\newtheorem{thm}{Theorem}
\newtheorem{lem}[thm]{Lemma}

\newenvironment{proof}[1][Proof]{\begin{trivlist}
\item[\hskip \labelsep {\bfseries #1}]}{\end{trivlist}}

\newcommand{\qed}{\nobreak \ifvmode \relax \else
      \ifdim\lastskip<1.5em \hskip-\lastskip
      \hskip1.5em plus0em minus0.5em \fi \nobreak
      \vrule height0.75em width0.5em depth0.25em\fi}

\begin{document}
%
\title{Learning-Based Distributed Detection-Estimation in Sensor Networks with Unknown Sensor Defects}

\author{Qing Zhou,
        Di~Li,~\IEEEmembership{Student Member,~IEEE,}
        Soummya~Kar,~\IEEEmembership{Member,~IEEE,}
        Lauren Huie, ~~~~~~~~~~~~~~~~~~~~~~~~~~~~~
        H. Vincent Poor,~\IEEEmembership{Fellow,~IEEE,}
        Shuguang Cui,~\IEEEmembership{Fellow,~IEEE}
        \thanks{Q. Zhou is with Qualcomm Inc., San Diego, CA 92121 USA (e-mail:zhouqing02@gmail.com).}
        \thanks{D. Li and S. Cui are with the Department of Electrical and Computer Engineering, Texas A\&M University, College Station, TX 77843 USA (e-mail: dili@tamu.edu; cui@tamu.edu). S. Cui is also a Distinguished Adjunct Professor at King Abdulaziz University in Saudi Arabia.}
        \thanks{S. Kar is with the Department of Electrical and Computer Engineering, Carnegie Mellon University, Pittsburgh, PA 15213 USA (e-mail: soummyak@andrew.cmu.edu).}
        \thanks{L. Huie is with the Air Force Research Lab, Rome, NY 13441 USA (e-mail:Lauren.Huie@rl.af.mil). }
       \thanks{H. V. Poor is with the Department of Electrical Engineering, Princeton University, Princeton, NJ 08544 USA (e-mail: poor@princeton.edu).}
       \thanks{Part of the work has been presented in GLOBECOM'11 and ICASSP'12 \cite{Qing2011,Qing2012}.}
        }
\maketitle

\begin{abstract}
We consider the problem of distributed estimation of an unknown deterministic scalar parameter (the target signal) in a wireless sensor network (WSN), where each sensor receives a single snapshot of the field. We assume that the observation at each node randomly falls into one of two modes: a valid or an invalid observation mode. Specifically, mode one corresponds to the desired signal plus noise observation mode (\emph{valid}), and mode two corresponds to the pure noise mode (\emph{invalid}) due to node defect or damage. With no prior information on such local sensing modes, we introduce a learning-based distributed procedure, called the mixed detection-estimation (MDE) algorithm, based on iterative closed-loop interactions between mode learning (detection) and target estimation. The online learning step re-assesses the validity of the local observations at each iteration, thus refining the ongoing estimation update process. The convergence of the MDE algorithm is established analytically. Asymptotic analysis shows that, in the high signal-to-noise ratio (SNR) regime, the MDE estimation error converges to that of an ideal (centralized) estimator with perfect information about the node sensing modes. This is in contrast to the estimation performance of a naive average consensus based distributed estimator (without mode learning), whose estimation error blows up with an increasing SNR.

\end{abstract}

\section{Introduction}\label{sec:Introduction}




\IEEEPARstart{A} key issue in wireless sensor network (WSN) design is to attain a meaningful network-wide consensus on knowledge based on unreliable locally sensed data \cite{Qing2012,tsitsiklisphd84, olfatisaberfaxmurray07, dimakiskarmourarabbatscaglione-11}. Due to the limited sensing capability and other unpredictable physical factors, such local observations may be invalid. For each single sensor, without jointly analyzing its observation with the other nodes, the validity of the data is not detectable. The traditional solution is to fuse data at a special powerful node named the fusion center. By collecting the data from all of the sensors, the fusion center could make a jointly optimal decision. If such a centralized solution is not possible, the distributed sensing problem arises \cite{Di2015,DiIcassp,Qing2011,Borkar82, deGroot74,Leger10,Ramanan12,Bertrand12,Mateos12}, where each sensor exchanges its local data with the neighbors, and merges the new information to its local estimate, in order to achieve the estimation accuracy of a centralized counterpart \cite{Estrin01, Rabbat07, Leilei10}. The existing research literature on relevant network-based distributed estimation may be broadly categorized into three classes. The first intensively studied family of distributed sensing problems consists of the so-called distributed network consensus or agreement problems and its variants \cite{Borkar82, deGroot74, Varagnolo10, Baquero12,Xiao0000}, of which a popular type is the distributed averaging problem, where a group of agents want to compute a liner function of a set of values distributed across the agent network, in particular, the average of their observations \cite{Boyd, Dimakis08}.
The second well-studied family of distributed sensing problems consists of distributed/decentralized estimation of parameters/processes in collaborative multi-agent networks with a single snapshot of the field, i.e., each agent obtains a single real or vector valued observation of the field at the beginning and no new observations are sampled over time. For example, in \cite{Olfati12,Li12} the authors studied estimation in static networks, where the sensors take a single snapshot of the field and then initiate distributed optimization to fuse the local estimates. The third well-studied family of distributed sensing problems consists of general time-sequential distributed estimation procedures for parameter inference in multi-agent networks in which agents access time-series observation data sequentially over time. In this family, two main approaches were proposed: the so-called consensus+innovation approach \cite{Kar_cooperation,Kar_QDlearning,Kar_efficiency} and the diffusion approach \cite{diffusion1,diffusion2,diffusion3}. We also mention the important and relevant literature on distributed detection and classification in multi-agent networks such as those based on the running
consensus approach \cite{17,18}, the diffusion approach \cite{19,20,21}, the consensus+innovtaions approach \cite{22,23,24}, and also~\cite{shahrampour2015switching, nedic2015fast, Sahu-SPRT-TSP-15, Srivastava-Leonard}.

In this paper, different from prior distributed approaches which focus solely on estimation or detection, we propose a mixed distributed detection-estimation algorithm with online interactions between detection and estimation. We assume that the observation process at each node randomly falls into one of the two modes, i.e., a valid observation mode vs. an invalid observation mode, where the valid observation is the desired signal plus noise and the invalid observation is just the pure noise. The rational behind this stochastic observation model is that the sensor might be damaged during deployment or physically blocked by certain objects between the sensor and the target; but the communication part in the sensor node still works. In this case, the sensor cannot observe a valid observation but pure noise, and the sensor itself cannot detect the validity of the observation on its own and keep executing the standard procedure in the network as a normal node. With the above setup, the traditional distributed consensus algorithms \cite{Wiesel12, Kar12_1, Kar11} could reach a naive averaging estimate of the target signal, by locally averaging the neighbor observations. However, the stochastic property of the observation modes may cause unreliable performance as shown later in the paper.

To address the above issue, a mixed detection-estimation (MDE) algorithm is introduced in this paper, which is a learning-based distributed procedure with closed-loop iterative interactions between the distributed mode learning and target estimation. In the MDE algorithm, the mode learning part detects the validity of the local observation iteratively as it performs the distributed estimation task. In each round of iteration, each node locally detects the observation validity with the maximum a posteriori probability (MAP) criterion based on the knowledge of the local current estimate of the target together with the local observation. The local estimate is then refined with the detected validities of the local observations and other exchanged information from the neighbors using a consensus + innovations type mechanism. By alternatively detecting validity and estimating the target, the sensor network can achieve a global consensus among all nodes. We analytically establish the convergence of the MDE algorithm. With asymptotic performance analysis, we show that in the high SNR regime, the local detection error on the observation mode converges to zero and the MDE estimation error converges to that of an ideal estimator with perfect information about the node defect status. The adaptive learning property of the MDE algorithm achieves a reliable estimation performance, in contrast to the unsatisfactory estimation performance of a naive average consensus based algorithm in the high SNR regime.




The rest of this paper is organized as follows. In Section \ref{sec:Network_Model}, we describe the network model first, and then present a naive averaging based estimation scheme and an ideal centralized estimation scheme as benchmarks. Section \ref{sec:Distributed_Algorithm} presents the MDE algorithm. Section \ref{sec:Main_Results} summarizes the main results, with some intermediate results proved in Section \ref{sec:Inter_Results}. Section \ref{sec:Proof_Con} formally proves the convergence of MDE. In Section \ref{sec:Performance}, we further analyze the performance of MDE, with some asymptotic analysis established in Section \ref{sec:Asym_Analysis}. Simulation results are presented in Section \ref{sec:Simulation_and_Discussion}. Finally, Section \ref{sec:Conclusion} concludes the paper.

\section{Network Model}\label{sec:Network_Model}
Let $\mathcal{N}_i$ and $\Omega_{i}$, $i \in \{1,2, ..., n\}$, denote sensor node $i$ and the set of its neighbors respectively. The received signal at $\mathcal{N}_i$ is $y_i=h_i \theta +w_i$, where $h_i \in \{0,1\}$ is an unknown validity index of the observation at node $\mathcal{N}_i$: i.e., $h_i=1$ indicates that $y_i$ is a valid observation and $h_{i}=0$ indicates the invalid observation case. In addition, $w_i$'s are independent Gaussian white noises with zero mean and variance $\sigma^2$. Although the exact instantiations of the $h_{i}$'s are unknown, we assume that $h_i$'s are i.i.d. Bernoulli random variables and the probability $p_1 \triangleq \Pr\{h_i=1\}$ is known \emph{a priori}. We denote the variance of $h_i$ as $\sigma_h^2$, i.e., $\sigma_h^2 = p_1(1-p_1)$. We are interested in estimating $\theta$ using an iterative distributed procedure, in which each node $\mathcal{N}_i$ may only use its neighbors' current state information for updating its local estimate (state) at time $t$. We assume that $\theta$ is a deterministic unknown target of real value. 

Denote by $\textbf{y}$ the network observation vector, i.e.,
\begin{equation}
\textbf{y}=\textbf{h} \theta + \textbf{w},
\end{equation}
with $\textbf{y}=[y_1, y_2, ..., y_n]^{T}$, $\textbf{h}=[h_1, h_2, ..., h_n]^{T}$, and $\textbf{w}=[w_1, w_2, ..., w_n]^{T}$. With this observation model, the sufficient statistic for estimation is $\textbf{y}$, and the optimal estimator is given by a maximum a posteriori (MAP) estimator. However, the complexity of MAP is too high to implement in practice. In order to reduce the complexity, we consider the linear estimator model. We note that, a straight-forward approach based on naive averaging could be cast as
\begin{eqnarray}
\hat{\theta}_\text{Naive} = \frac{\textbf{1}^T \textbf{y}}{n p_1}\label{equ:x_naive},
\end{eqnarray}
which yields a linear minimum variance unbiased estimator (LMVUE) with the property that $\hat \theta_\text{Naive} \rightarrow \theta$ almost surely as $n \rightarrow \infty$. The variance (which coincides with the mean-squared error) of $\hat{\theta}_\text{Naive}$ may be expressed as
\begin{equation}
 \operatorname{Var}(\hat \theta_\text{Naive}) 
                            = \frac{1}{n}\left[\frac{\sigma^2}{p_1^2} (1+\mbox{SNR} \sigma_h^2)\right],\label{equ:var_x_naive}
\end{equation}
where $\text{SNR}$ is defined as $\frac{\theta^2}{\sigma^2}$.

Although this naive estimate is quite straight-forward in terms of implementation \cite{olfatisaberfaxmurray07,Boyd,Dimakis08}, we observe from~\eqref{equ:var_x_naive} that the precision is poor in the high SNR regime, where in particular, the mean-squared error (MSE) blows up with an increasing SNR. On the other extreme, if we assume that $\textbf{h}$ is perfectly known, we may generate an \emph{ideal} estimate $\hat \theta_\text{Ideal}$ of $x$ by eliminating the invalid observations, i.e.,
\begin{equation}\label{equ:theta_ideal}
\hat \theta_\text{Ideal} = \frac{\sum_{\{i: h_i=1\}} y_i}{\sum_{\{i:h_i=1\}} h_i} = \frac{\sum_{i=1}^n h_i y_i}{\sum_{i=1}^n h_i}.
\end{equation}
The above estimate is also unbiased, with $\hat \theta_\text{Ideal}\rightarrow \theta$ almost surely as $n \rightarrow \infty$, and its variance may be expressed as
\begin{align}\label{equ:var_ideal}
\operatorname{Var}(\hat \theta_\text{Ideal}) &= \operatorname{E} (\operatorname{Var}(\hat \theta_\text{Ideal}|\textbf{h})) \!+\! \operatorname{Var}(\operatorname{E}(\hat \theta_\text{Ideal}|\textbf{h}))= \psi \sigma^2,
\end{align}
where $\psi=\sum_{k=1}^n \frac{1}{k} \binom{n}{k} p_1^k (1-p_1)^{n-k}$, the derivation of which is given in Appendix \ref{Var_ideal}. We note that $\psi$ is not related to $\text{SNR}$ and is on the order of $\frac{1}{n}$. For example, when $p_1=0.5$, we have $\psi \approx \frac{2-2^{-n}}{n+1}$. A key difference from the naive estimate in (\ref{equ:x_naive}) is that the variance of the ideal estimate stays constant over SNR, i.e., the estimation error does not scale up with the SNR.

From the MSE viewpoint, the ideal estimate is in fact optimal as long as the observation noise is Gaussian. However, such a scheme may not be implementable as it requires the perfect knowledge of $\textbf{h}$, which is unknown \emph{a priori}. In Section \ref{sec:Distributed_Algorithm}, we introduce a learning-based distributed estimation procedure, the MDE algorithm, based on the iterative detection of $\textbf{h}$ and estimate refinement of $\theta$. Our results indicate that not only $\textbf{h}$ could be detected with high accuracy by the MDE algorithm, but also does the estimation performance (in terms of MSE) approach that of the ideal estimate $\hat \theta_\text{Ideal}$ in the high SNR regime.


\section{Distributed MDE Algorithm}\label{sec:Distributed_Algorithm}
In this section, we present the MDE algorithm for the problem of interest. 
In each iteration of the MDE algorithm, each node first locally detects the value of $h_{i}$ by using its current local estimate of $\theta$ and its local observation. This initially detected observation validity index is used to update some intermediate parameters, which are subsequently forwarded to the neighboring nodes. This leads to an estimate refinement process, which feeds back new information to improve the validity detection in the next iteration. The algorithm at sensor $i$ is presented as follows.

Step 1. Initialization at time $1$
\begin{eqnarray}\label{equ:iteration_start}
\hat \theta_i^{+} (1) = \hat \theta_i^{-} (1) = \frac{\sum_{j \in \Omega_i} y_j}{|\Omega_i| p_1},~~ \hat {\bar y}_i (1) = y_i.
\end{eqnarray}

Step 2. Detection of $h_{i}$ at time $t>1$
\begin{eqnarray}
(\hat \theta_i^{+}(t))^2 - 2 y_i \hat \theta_i^{+} (t) \overset{\hat h_i^{+}(t)=1 }{\underset{\hat h_i^{+}(t)=0}\lesseqgtr} 2 \sigma^2 \ln \frac{p_1}{p_0}, \label{equ:detection1}\\
(\hat \theta_i^{-}(t))^2 - 2 y_i \hat \theta_i^{-} (t) \overset{\hat h_i^{-}(t)=1}{\underset{\hat h_i^{-}(t)=0}\lesseqgtr} 2 \sigma^2 \ln \frac{p_1}{p_0}, \label{equ:detection}
\end{eqnarray}
where $p_0=1-p_1$.


Step 3. Calculation of intermediate parameters $u$, $v$, and the estimation of $\bar y$ at time $t$
\begin{align}
u_i^{+}(t) &= u_i^{+}(t-1) -\beta(t) \sum\nolimits_{j \in \Omega_i} \big(u_i^{+}(t-1)-u_j^{+}(t-1)\big) \nonumber\\
&+ \alpha(t) \big( y_i \hat h_i^{+}(t) - u_i^{+}(t-1)\big),\label{equ:u_i_t_+} \\
v_i^{+}(t) &= v_i^{+}(t-1) -\beta(t) \sum\nolimits_{j \in \Omega_i} \big(v_i^{+}(t-1)-v_j^{+}(t-1)\big) \nonumber\\
& + \alpha(t) \big( \hat h_i^{+}(t) - v_i^{+}(t-1)\big),~~~~\label{equ:v_i_t_+} \\
u_i^{-}(t) &= u_i^{-}(t-1) -\beta(t) \sum\nolimits_{j \in \Omega_i} \big(u_i^{-}(t-1)-u_j^{-}(t-1)\big) \nonumber\\
&+ \alpha(t) \big( y_i \hat h_i^{-}(t) - u_i^{-}(t-1)\big),\label{equ:u_i_t_-} \\
v_i^{-}(t) &= v_i^{-}(t-1) -\beta(t) \sum\nolimits_{j \in \Omega_i} \big(v_i^{-}(t-1)-v_j^{-}(t-1)\big) \nonumber\\
&+ \alpha(t) \big( \hat h_i^{-}(t) - v_i^{-}(t-1)\big),~~~~\label{equ:v_i_t_-} \\
\hat{\bar y}_{i}(t) &= \hat{\bar y}_{i}(t-1) - \beta(t)\sum\nolimits_{j\in\Omega_{i}}(\hat{\bar y}_{i}(t-1) - \hat{\bar y}_{j}(t-1)),\label{equ:y}
\end{align}
where $u_i^{+}(0) = y_i \hat h_i^{+}(1)$, $v_i^{+} (0) = \hat h_i^{+}(1)$, $u_i^{-}(0) = y_i \hat h_i^{-}(1)$, $v_i^{-} (0) = \hat h_i^{-}(1)$, and $\alpha(t)$ and $\beta(t)$ satisfy the following four conditions:
\begin{center}
\begin{itemize}
\item $0<\alpha(t)<1$ and $0<\beta(t)<1$,
\item $\alpha(t) \rightarrow 0$, $\beta(t) \rightarrow 0$,
\item $\beta(t)/\alpha(t) \rightarrow \infty$,
\item $\sum_{t=1}^\infty \alpha(t) = \infty$, $\sum_{t=1}^\infty \beta(t) = \infty$.
\end{itemize}
\end{center}


Step 4. Estimation update of $\theta$

\begin{equation}
\hat \theta_i (t+1) = \left\{
   \begin{array}{c}
   \hat \theta_i^{+} (t+1),~~~~~\hat {\bar y}_i(t) \geq 0 \\
   \hat \theta_i^{-} (t+1),~~~~~\hat {\bar y}_i(t) < 0  \\
   \end{array}
  \right.
\end{equation}
where $\hat \theta_i^{+} (t+1) = \max\left\{\frac{u_i^{+}(t)}{v_i^{+}(t)+\delta}, 0 \right\}$ and $\hat \theta_i^{-} (t+1) = \min\left\{\frac{u_i^{-}(t)}{v_i^{-}(t)+\delta}, 0 \right\}$, with $\delta$ as an arbitrary small positive constant, to prevent the denominator from being zero.

We then repeat steps 2 to 4 until $\left|\frac{u_i^{+}(t)}{v_i^{+}(t)+\delta}-\frac{u_i^{+}(t-1)}{v_i^{+}(t-1)+\delta}\right|<\epsilon$, $\left|\frac{u_i^{-}(t)}{v_i^{-}(t)+\delta}-\frac{u_i^{-}(t-1)}{v_i^{-}(t-1)+\delta}\right|<\epsilon$, and $|\hat{\bar y}_i(t) -\hat{\bar y}_i(t-1)|<\epsilon$, $\forall i$,  where $\epsilon$ is a predefined small positive error tolerant parameter.

Basically, the algorithm starts with a linear minimum variance unbiased estimator (LMVUE) among 1-hop neighbors as the initial estimator in step 1. In step 2, each node locally detects (re-assesses) the value of $h_i$ using the current local estimate of $\theta$ and $y_i$. The validity indices, thus obtained, are used to update intermediate parameters that are subsequently forwarded to the neighboring nodes, leading to the state update in step 3, where each node refines its local parameters, i.e., $u_i^+(t)$, $u_i^-(t)$, $v_i^+(t)$, and $v_i^-(t)$, based on new information from its neighbors using a consensus + innovations type mechanism. (The consensus potential governs how neighboring observations are assimilated to seek agreement among agents, whereas, the local innovation potential may be viewed as a refinement capturing the agent's local observation and its instantaneous validity measure.) Finally, a new estimate is obtained from $u_i^+(t)$, $u_i^-(t)$, $v_i^+(t)$, and $v_i^-(t)$, and a new iteration starts if needed. In the next section, we investigate the convergence of this iterative procedure. We also emphasize that the conditions on $\alpha(t)$ and $\beta(t)$ listed above are not hard to satisfy. For example, we may choose $\alpha(t)=\delta_a/t$, and $\beta(t)=\delta_b/t^{1-\varepsilon}$, with $\varepsilon \in (0,1)$, $\delta_a$ and $\delta_b$ as small positive real constants.

\section{Main Results}\label{sec:Main_Results}
In this section, we present the main results, with the proofs given in the subsequent sections.  

\begin{thm}\label{thm:hat_x_t}
Let the inter-sensor communication network be connected\footnote{The network is said to be connected if there exists a path (possibly multi-hop) between any pair of nodes.}, and assume that $\alpha(t)$ and $\beta(t)$ in~\eqref{equ:u_i_t_+}-\eqref{equ:y} satisfy the following four conditions:
\begin{itemize}
\item $0 < \alpha(t) <1$ and $0 < \beta(t) <1$,
\item $\alpha(t) \rightarrow 0$, $\beta(t) \rightarrow 0$,
\item $\beta(t)/\alpha(t) \rightarrow \infty$,
\item $\sum_{t=1}^\infty \alpha(t) = \infty$, $\sum_{t=1}^\infty \beta(t) = \infty$.
\end{itemize}
Then, the estimate sequence $\{\hat \theta_i (t)\}$ at each node $\mathcal{N}_{i}$ converges almost surely as
\begin{align}\label{equ_thm1}
&\lim_{t \rightarrow \infty} \hat \theta_i (t) = \nonumber\\
&\left\{
   \begin{array}{c}
   \max\left\{\frac{\sum_{j=1}^n \hat h_j^{+} y_j}{\sum_{j=1}^n \hat h_j^{+} + n \delta},0\right\},~\mbox{on~the~event}~\{\bar{y}\geq 0\} \\
   \min\left\{\frac{\sum_{j=1}^n \hat h_j^{-} y_j}{\sum_{j=1}^n \hat h_j^{-} + n \delta},0\right\},~\mbox{on~the~event}~\{\bar {y} < 0\}  \\
   \end{array}, \forall i,
  \right.
\end{align}
where $\hat h_i^{(\cdot)} \in \{0,1\}$ denotes the limiting value of the convergent sequence $\{\hat h_i^{(\cdot)}(t)\}$, in which we use $(\cdot)$ to denote either $+$ or $-$; $\bar y$ is the arithmetic mean of all $y_i$'s. Note that $\hat h_{i}^{(\cdot)}$ is, in general, random given the stochasticity of the $h_{i}$s and the $y_{i}$'s.
\end{thm}

The proof of Theorem \ref{thm:hat_x_t} is presented in Section \ref{sec:Proof_Con}. Theorem \ref{thm:hat_x_t} shows that the estimate sequence $\{\hat \theta_i (t)\}$ at each node converges to a unique (stochastic) limit, denoted by $\hat \theta_{\mbox{\scriptsize Ideal}}$, as $t\rightarrow\infty$, which implies that the nodes in the network achieve agreement over the estimate of the unknown parameter $\theta$, i.e., realizing the network consensus. Since we consider a general real valued parameter $\theta$, according to the proposed algorithm, the limiting estimate value takes on different forms depending on whether the event $\{\bar{y}\geq 0\}$ or its complement holds, reflecting the possible non-negativity or negativity of the parameter $\theta$ respectively. We further prove in Theorem \ref{thm:E_theta_to_theta} that this converged estimation value is unbiased in the asymptotic regime as $\text{SNR}$ goes to infinity.

\begin{thm}\label{thm:h_1>h_2}
If we order the observations $\{y_i\}$ in the increasing order as $y_{(1)}\leq y_{(2)}\leq ...\leq y_{(n)}$, and denote the corresponding decisions given in step 2 of the proposed algorithm as $\hat h_{(1)}^{(\cdot)}, \hat h_{(2)}^{(\cdot)}, ..., \hat h_{(n)}^{(\cdot)}$, we have
\begin{equation}
\hat h_{(1)}^+ \leq \hat h_{(2)}^+ \leq ... \leq \hat h_{(n)}^+, 
\end{equation}
and
\begin{equation}
\hat h_{(1)}^- \geq \hat h_{(2)}^- \geq ... \geq \hat h_{(n)}^-, 
\end{equation}
where $\hat h_{(i)}^{(\cdot)} \in \{0,1\}$.
\end{thm}

We prove Theorem \ref{thm:h_1>h_2} in Section~\ref{subsec:shrink}. Theorem \ref{thm:h_1>h_2} demonstrates an interesting property of the proposed algorithm: if the observations from different nodes are ordered, the corresponding $\hat h_{(i)}^{(\cdot)}$'s are also ordered. Specifically, if the observations are increasingly ordered, $\hat h_{(i)}^{+}$'s have the same increasing order as that of observations, while $\hat h_{(i)}^-$'s inherit a decreasing order. Since $\hat h_{(i)}^+$'s correspond to \eqref{equ:detection1} with non-negative $\theta_i^{+} (t)$ and $\hat h_{(i)}^-$'s correspond to \eqref{equ:detection} with non-positive $\theta_i^{-} (t)$, this intuitively explains why $\hat h_{(i)}^+$'s and $\hat h_{(i)}^-$'s have different orders.

\begin{thm}\label{thm:E_theta_to_theta}
For the MDE algorithm, we have
\begin{eqnarray}
\lim_{\mbox{\scriptsize  SNR} \rightarrow \infty} E(\hat \theta) = \theta,
\end{eqnarray}
where $\hat \theta$ is the converged value shown in \eqref{equ_thm1}. Since the converged value in \eqref{equ_thm1} does not depend on the node index $i$, the index is dropped.
\end{thm}

The proof of Theorem \ref{thm:E_theta_to_theta} is presented in Section~\ref{subsec:E_hat_theta}. Theorem \ref{thm:E_theta_to_theta} shows that the converged estimation value in \eqref{equ_thm1} is unbiased in the asymptotic regime as SNR$\rightarrow \infty$.

%
%

\begin{thm}\label{thm:var_theta->Ideal}
For the MDE algorithm, we have
\begin{equation}
\lim_{n\rightarrow \infty,\mbox{\scriptsize SNR} \rightarrow \infty}\operatorname{Var}(\hat \theta) = \operatorname{Var}(\hat \theta_{\mbox{\scriptsize Ideal}}),
\end{equation}
where $\hat \theta_{\mbox{\scriptsize Ideal}}$ is the ideal estimator defined in \eqref{equ:theta_ideal}.
\end{thm}

The proof of Theorem \ref{thm:var_theta->Ideal} is given in Section~\ref{sec:Asym_Analysis}. Theorem \ref{thm:var_theta->Ideal} shows that the estimation error variance converges almost surely to that of the ideal estimate $\hat \theta_\text{Ideal}$ defined in \eqref{equ:theta_ideal}, when both node number $n$ and SNR increase.

By combining Theorem \ref{thm:E_theta_to_theta} and Theorem \ref{thm:var_theta->Ideal}, we see that the performance of our proposed distributed algorithm converges to that of the ideal estimate $\hat \theta_\text{Ideal}$ defined in \eqref{equ:theta_ideal}. Since this ideal estimate is computed based on the assumption that $\textbf{h}$ is perfectly known or precisely learned, as an optimal estimation method, its performance is the benchmark of all other estimation algorithms to deal with unknown sensor defects. Theorem \ref{thm:E_theta_to_theta} and Theorem \ref{thm:var_theta->Ideal} imply that the proposed distributed algorithm converges to the optimal solution and the validity index $\textbf{h}$ can be precisely learned, as SNR goes to infinity.

\section{Intermediate Results for Proofs}\label{sec:Inter_Results}
In this section, we establish some intermediate results to be used later. In the MDE algorithm, we note that the positive and negative parts are symmetric, i.e., $\hat \theta^{+}_i (t)$ vs. $\hat \theta^{-}_i (t)$, $h_i^+ (t)$ vs. $h_i^- (t)$, $u_i^+(t)$ vs. $u_i^-(t)$, and $v_i^+(t)$ vs. $v_i^-(t)$. In the following, we use $(\cdot)$ to denote either $+$ or $-$ and the results can be applied to both of these two cases. We denote $\frac{\sum_i u_i^{(\cdot)}(t)}{n}$ and $\frac{\sum_i v_i^{(\cdot)}(t)}{n}$ as $\bar u^{(\cdot)}(t)$ and $\bar v^{(\cdot)}(t)$, respectively. In the following, Lemma \ref{lem:y_avg_upper_bounded} proves that $\bar u^{(\cdot)}(t)$ is a bounded sequence. Then we show the limiting relationship between $\bar u^{(\cdot)}(t)$ and $u_i^{(\cdot)}(t)$ in Lemma \ref{lem:u_i-y_avg_0}, where $\lim_{t \rightarrow \infty} \big(u_i^{(\cdot)}(t)-\bar u^{(\cdot)}(t)\big) =0$.
 Both $u_i^{(\cdot)}(t)$ and $\bar u^{(\cdot)}(t)$ in the above results could be replaced by $v_i^{(\cdot)}(t)$ and $\bar v^{(\cdot)}(t)$ respectively and the proofs are similar. Then, Lemma \ref{lem:y_avg_v_avg_0} proves that $\lim_{t \rightarrow \infty} \big(\frac{\bar u^{(\cdot)}(t+1)}{\bar v^{(\cdot)}(t+1)+\delta}-\frac{\bar u^{(\cdot)}(t)}{\bar v^{(\cdot)}(t)+\delta}\big) =0$. After that, the limiting relationship between $\hat \theta_i(t)$ and $\frac{\bar u^{(\cdot)}(t)}{\bar v^{(\cdot)}(t)+\delta}$ is proved in Lemma \ref{lem:x_y_avg_v_avg_0}.

\begin{lem}\label{lem:y_avg_upper_bounded}
Let the inter-sensor communication network be connected. Thus we have that $\bar u^{(\cdot)}(t)$ is a bounded sequence.
\end{lem}
\begin{proof}
In step 3 of the algorithm, we have
\begin{eqnarray}
u_i^{(\cdot)}(t)=u_i^{(\cdot)}(t-1) + \alpha(t) \big( y_i \hat h_i^{(\cdot)}(t) - u_i^{(\cdot)}(t-1)\big) \nonumber\\
-\beta(t) \sum\nolimits_{j \in \Omega_i} \big(u_i^{(\cdot)}(t-1)-u_j^{(\cdot)}(t-1)\big).
\end{eqnarray}

Taking the average on both sides over all $i \in [1,\cdots,n]$, we have the iterative expression of $\bar u^{(\cdot)}(t)$ as follows
\begin{align}
\bar u^{(\cdot)}(t)&= \nonumber\\
&\bar u^{(\cdot)}(t-1) + \alpha(t) \big( \{y_i \hat h_i^{(\cdot)}(t)\}_{avg} - \bar u^{(\cdot)}(t-1)\big),
\end{align}
where $\{y_i \hat h_i^{(\cdot)}(t)\}_{avg} = \sum_{i=1}^n (y_i \hat h_i^{(\cdot)}(t))/n$.

We rewrite the above equation in another form as
\begin{equation}
\bar u^{(\cdot)}(t)=\big(1-\alpha(t)\big) \bar u^{(\cdot)}(t-1) + \alpha(t) \{y_i \hat h_i^{(\cdot)}(t)\}_{avg}; \label{equ:y_avg_t}
\end{equation}
and for $\bar u^{(\cdot)}(t+1)$, we have
\begin{align}
\bar u^{(\cdot)}(t+1)&=\big(1-\alpha(t+1)\big) \bar u^{(\cdot)}(t) + \alpha(t+1) \nonumber\\
 &\{y_i \hat h_i^{(\cdot)}(t+1)\}_{avg}. \label{equ:bar_u_t+1}
\end{align}
By substituting (\ref{equ:y_avg_t}) into the right-side of \eqref{equ:bar_u_t+1}, we have
\begin{align}
&|\bar u^{(\cdot)}(t+1)|= \nonumber\\
&\left|\big(1\!-\!\alpha(t+1)\big)\big[  \big(1\!-\!\alpha(t)\big) \bar u^{(\cdot)}(t-1) \!+\! \alpha(t) \{y_i \hat h_i^{(\cdot)}(t)\}_{avg} \big] \right.\nonumber\\
&\left.+ \alpha(t+1) \{y_i \hat h_i^{(\cdot)}(t+1)\}_{avg}\right|\nonumber\\
 \end{align}
 \begin{align}           
            &= \left|\prod_{j = t-1}^{t} \big(1-\alpha(j+1)\big) \bar u^{(\cdot)}(t-1) + \big(1-\alpha(t+1)\big)\right.\nonumber\\
            &\left.\alpha(t) \{y_i \hat h_i^{(\cdot)}(t)\}_{avg} + \alpha(t+1) \{y_i \hat h_i^{(\cdot)}(t+1)\}_{avg}\right|\nonumber\\
            &\leq \prod_{j = t-1}^{t} \big(1-\alpha(j+1)\big)  |\bar u^{(\cdot)}(t-1)| + \big(1-\alpha(t+1)\big)\nonumber\\
            &\alpha(t) y_{max} + \alpha(t+1) y_{max}\nonumber\\
            &= \prod_{j = t-1}^{t}\!\big(1-\alpha(j+1)\big)\big(|\bar u^{(\cdot)}(t-1)|-y_{max}\big) \!+ \!y_{max},
\end{align}
where $y_{max}=\max_{i=1}^n |y_i|$ is the natural upper bound of $\left|\{y_i \hat h_i^{(\cdot)}(t)\}_{avg}\right|$.
Iteratively, we deduce that
\begin{equation}
|\bar u^{(\cdot)}(t+1)| \leq  \prod_{j = 1}^{t} \big(1-\alpha(j+1)\big)\big(|\bar u^{(\cdot)}(1)|-y_{max}\big)  + y_{max}.
\end{equation}

Note that $1-a \leq e^{-a}$ for $0 \leq a \leq 1$; thus we have
\begin{eqnarray}
            & \prod_{i = 1}^{t} \big(1-\alpha(i+1)\big)\big(|\bar u^{(\cdot)}(1)| - y_{max}\big) + y_{max}\nonumber\\
            &\leq e^{-\sum_{i=1}^{t} \alpha(i+1)} \big( |\bar u^{(\cdot)}(1)| - y_{max} \big) + y_{max}.
\end{eqnarray}
When $t \rightarrow \infty$, we have $e^{-\sum_{i=0}^{t} \alpha(i+1)} \rightarrow 0$ by the fourth condition of $\alpha(t)$; and then we conclude that $\bar u^{(\cdot)}(t)$ is a bounded function.
\end{proof}

\begin{lem} \label{lem:u_i-y_avg_0}
Let the inter-sensor communication network be connected. We have
\begin{eqnarray}
\lim_{t \rightarrow \infty} \big(u_i^{(\cdot)}(t) - \bar u^{(\cdot)}(t)\big) = 0, ~\forall i \nonumber\\
\lim_{t \rightarrow \infty} \big( v_i^{(\cdot)}(t)- \bar v^{(\cdot)}(t)\big) = 0, ~\forall i \nonumber
\end{eqnarray}
with $\bar u^{(\cdot)}(t)$ and $\bar v^{(\cdot)}(t)$ defined previously.
\end{lem}
\begin{proof}
This Lemma can be proved by applying Lemma 15 in \cite{Kar11}, which is skipped here.
\end{proof}

In Lemmas \ref{lem:y_avg_upper_bounded} and \ref{lem:u_i-y_avg_0}, $u_i^{(\cdot)}(t)$ and $\bar u^{(\cdot)}(t)$ could be directly replaced by $v_i^{(\cdot)}(t)$ and $\bar v^{(\cdot)}(t)$ respectively, and the proofs are similar.

\begin{lem}\label{lem:y_avg_v_avg_0}
Let the inter-sensor communication network be connected. Then,
\begin{equation}
\lim_{t \rightarrow \infty} \left(\frac{\bar u^{(\cdot)}(t+1)}{\bar v^{(\cdot)}(t+1)+\delta}-\frac{\bar u^{(\cdot)}(t)}{\bar v^{(\cdot)}(t)+\delta}\right) = 0.
\end{equation}
\end{lem}
\begin{proof}
We have
\begin{align}
&\frac{\bar u^{(\cdot)}(t+1)}{\bar v^{(\cdot)}(t+1)+\delta}-\frac{\bar u^{(\cdot)}(t)}{\bar v^{(\cdot)}(t)+\delta}\nonumber\\
 &= \frac{\big(1-\alpha(t+1)\big)\bar u^{(\cdot)}(t) + \alpha(t+1) \{y_i \hat h_i^{(\cdot)}(t+1)\}_{avg}}{\big(1-\alpha(t+1)\big)\bar v^{(\cdot)}(t) + \alpha(t+1) \{\hat h_i^{(\cdot)}(t+1)\}_{avg}+\delta}\nonumber\\
 &- \frac{\bar u^{(\cdot)}(t)}{\bar v^{(\cdot)}(t)+\delta}\nonumber\\
 \end{align}
 \begin{align}
&= \left\{\frac{ \{y_i \hat h_i^{(\cdot)}(t+1)\}_{avg} (\bar v^{(\cdot)}(t) + \delta) - \bar u^{(\cdot)}(t) \delta}{\big(1-\alpha(t+1)\big)\bar v^{(\cdot)}(t) + \alpha(t) \{\hat h_i^{(\cdot)}(t+1)\}_{avg} +\delta   } \right.\nonumber\\
&\left.-\frac{ \{\hat h_i^{(\cdot)}(t+1)\}_{avg} \bar u^{(\cdot)}(t)}{ \big(1-\alpha(t+1)\big)\bar v^{(\cdot)}(t) + \alpha(t) \{\hat h_i^{(\cdot)}(t+1)\}_{avg} +\delta  } \right\}\nonumber\\
&\cdot \frac{\alpha(t+1)}{(\bar v(t)+\delta)}
\end{align}

By Lemma \ref{lem:y_avg_upper_bounded}, both $\bar v^{(\cdot)}(t)$ and $\bar u^{(\cdot)}(t)$ are bounded (both upper- and lower-bounded). In addition, $\{y_i \hat h_i^{(\cdot)}(t+1)\}_{avg}$ and $\{\hat h_i^{(\cdot)}(t+1)\}_{avg}$ are naturally bounded. Together with the fact that $\delta$ is an arbitrarily small constant, and $\lim_{t \rightarrow \infty} \alpha(t) = 0$, we conclude that $\lim_{t \rightarrow \infty} \left(\frac{\bar u^{(\cdot)}(t+1)}{\bar v^{(\cdot)}(t+1)+\delta}-\frac{\bar u^{(\cdot)}(t)}{\bar v^{(\cdot)}(t)+\delta}\right) = 0$.
\end{proof}

\begin{lem}\label{lem:x_y_avg_v_avg_0}
Let the inter-sensor communication network be connected. Then,
\begin{eqnarray}
\lim_{t \rightarrow \infty} \left(\hat \theta_i^{+}(t+1) - \max\left\{\frac{\bar u^{+}(t)}{\bar v^{+}(t)+\delta},0\right\}\right) = 0,~\forall i,\nonumber\\
\lim_{t \rightarrow \infty} \left(\hat \theta_i^{-}(t+1) - \min\left\{\frac{\bar u^{-}(t)}{\bar v^{-}(t)+\delta},0\right\}\right) = 0,~\forall i,\nonumber
\end{eqnarray}
where $\bar u^{(\cdot)}(t)$ and $\bar v^{(\cdot)}(t)$ denote the averaging values of $u_i^{(\cdot)}(t)$ and $v_i^{(\cdot)}(t)$, respectively.
\end{lem}
\begin{proof}
Recall $\hat \theta_i^{+}(t+1) = \max\left\{\frac{u_i^{+}(t)}{v_i^{+}(t)+\delta},0\right\}$ and $\hat \theta_i^{-}(t+1) = \min\left\{\frac{u_i^{-}(t)}{v_i^{-}(t)+\delta},0\right\}$. We have
\begin{align}
&  \lim_{t \rightarrow \infty} \left(\frac{u_{i}^{(\cdot)}(t)}{v_{i}^{(\cdot)}(t)+\delta} - \frac{\bar u^{(\cdot)}(t)}{\bar v^{(\cdot)}(t)+\delta}\right)=\nonumber\\
& \lim_{t \rightarrow \infty} \left(\frac{u_{i}^{(\cdot)}(t)\bar v^{(\cdot)}(t)-\bar u^{(\cdot)}(t)v_{i}^{(\cdot)}(t) +\delta(u_i^{(\cdot)}(t)-\bar u^{(\cdot)}(t))}{(v_{i}^{(\cdot)}(t)+\delta)(\bar v^{(\cdot)}(t)+\delta)} \right)\nonumber\\
&=0,\nonumber
\end{align}
which is according to Lemma \ref{lem:u_i-y_avg_0}.
Therefore, the proof is completed.
\end{proof}

\section{Proof of Theorem \ref{thm:hat_x_t}}\label{sec:Proof_Con}
In this section, we prove the convergence and derive the limiting value for Theorem \ref{thm:hat_x_t}. Without loss of generality, we prove the case of $\hat \theta_i^+(t)$ and skip the proof of $\hat \theta_i^-(t)$, which is similar. We first partition the real axis in Subsection \ref{subsec:partition}, such that the detection of $\mathbf{\hat h}^+(t)$ has the same results when $\hat \theta_i^+(t)$'s are in the same interval. Then we derive the smooth moving condition in Subsection \ref{subsec:smooth_moving_condition}, under which $\hat \theta_i^+(t)$ moves on the real axis by passing the partitions sequentially along the iteration process, such that the changing of $\mathbf{\hat h}^+(t)$ is successive with time. From the proposed algorithm, we notice that the local estimation is the greater one between $0$ and $\frac{u_i^+(t)}{v_i^+(t)+\delta}$, when $\hat {\bar y}_i(t) \geq 0$. As such, we only need to prove the convergence of $\frac{u_i^+(t)}{v_i^+(t)+\delta}$, and then the convergence of $\hat \theta_i^+(t)$ is guaranteed.
In Subsection \ref{subsec:MDE_algorithm}, we complete the proof of Theorem \ref{thm:hat_x_t}.

\subsection{Partitions of the Real Axis}\label{subsec:partition}
We now seek a suitable scale to study the iteration procedure. We start by exploring step 2 of the proposed algorithm. For each $i$, we make a hard decision at step 2. We define the region that returns $\hat h_i^+(t) = 1$ as the decision region of $\hat \theta_i^+ (t)$, denoted by $\mathcal{D}_i$. In particular, if $y_i^2 +2 \sigma^2 \ln \frac{p_1}{p_0}\ge0$, we have $\mathcal{D}_i = [r_{i}^-,r_{i}^+]$ for node $i$, where $r_{i}^- = y_i -\sqrt{y_i^2 +2 \sigma^2 \ln \frac{p_1}{p_0}},r_{i}^+= y_i +\sqrt{y_i^2 +2 \sigma^2 \ln \frac{p_1}{p_0}}$; otherwise, $\mathcal{D}_i = \emptyset$. Next we partition the real axis into at most $2n+1$ parts by these boundaries of $\mathcal{D}_i$'s, i.e., $r_i^-$'s and $r_i^+$'s. Here, we say ``at most" due to the fact that some of the $r_i^{\cdot}$'s may not exist, e.g., when $y_i^2 +2 \sigma^2 \ln \frac{p_1}{p_0} < 0$ or when multiple boundaries share the same value.  Then we name these boundaries in an increasing order of their values as $b_1$ to $b_M$ and name the partitioned left-open and right-closed intervals as $\mathcal{I}_1$ to $\mathcal{I}_{M+1}$, from left to right on the real axis.

\subsection{Smooth Moving Condition}\label{subsec:smooth_moving_condition}
In this subsection, we define the gathering region of $\{\hat \theta_i^+(t)\}$ as $\mathcal{G}^+(t)$, which is the range that covers all possible values of $\hat \theta_i^+(t)$'s. Then we study the condition for $\mathcal{G}^+(t)$ to move on the axis smoothly during the iteration process. In other words, the gathering region touches those boundaries $b_m$'s (from $\{b_1,\cdots,b_M\}$) sequentially in order without jumping if it passes through the boundaries. Also, for each time, the gathering region $\mathcal{G}^+(t)$ touches at most one of those different boundaries at each iteration. Next, we propose two conditions to guarantee the above situation.

We choose $\varepsilon$ that is at least 3 (the reason of choosing 3 is explained at the end of this subsection) times smaller than the narrowest range in $\mathcal{K}_j$'s, i.e., $3 \varepsilon < \min_{j}\{ |\mathcal{K}_j|\}$, where $\mathcal{K}_j$'s are the intervals partitioned jointly by $b_m$'s and $y_i$'s.  (such that the number of $\mathcal{K}_j$'s is larger, the minimum length of $\mathcal{K}_j$'s is shorter, than $\mathcal{I}_m$'s).
\begin{itemize}
\item By Lemma \ref{lem:x_y_avg_v_avg_0}, we have \[\lim_{t \rightarrow \infty} \left( \hat \theta_i^+(t) - \max\left\{\frac{\bar u^+(t-1)}{\bar v^+(t-1)+\delta},0\right\}\right) = 0.\] Thus we could find $t_\varepsilon$, such that for any $t>t_\varepsilon$, we have $\left|\hat \theta_i^+(t) - \max\left\{\frac{\bar u^+(t-1)}{\bar v^+(t-1)+\delta},0\right\}\right|< \varepsilon$.
\item  By Lemma \ref{lem:y_avg_v_avg_0}, we have \[\lim_{t \rightarrow \infty} \left(\frac{\bar u^+(t+1)}{\bar v^+(t+1) + \delta }-\frac{\bar u^+(t)}{\bar v^+(t)+\delta}\right) = 0.\] Thus we could find $t_\alpha$, such that for any $t>t_\alpha$, $\left|\frac{\bar u^+(t+1)}{\bar v^+(t+1)+\delta}-\frac{\bar u^+(t)}{\bar v^+(t)+\delta}\right| < \varepsilon$.
\end{itemize}
When $t> \max(t_\varepsilon, t_\alpha)$, we define $\mathcal{G}^+(t) = \left(\max\left\{\frac{\bar u^+(t-1)}{\bar v^+(t-1)+\delta},0\right\} - \varepsilon, \max\left\{\frac{\bar u^+(t-1)}{\bar v^+(t-1)+\delta},0\right\} + \varepsilon\right)$ as the gathering region of $\hat \theta_i^+(t)$, i.e., $\hat \theta_i^+(t) \in \mathcal{G}^+(t)$, $\forall i$. Since $\varepsilon < \frac{1}{3} \min_{j}\{ |\mathcal{K}_j|\}\leq \frac{1}{3} \min_{m}\{ |\mathcal{I}_m|\}$, $\mathcal{G}^+(t)$ does not touch or pass two successive $b_m$'s during two successive iterations as desired. In the sequel, we assume all the iterations under concern satisfy $t > \max(t_\varepsilon, t_\alpha)$.

\subsection{Proof of Theorem~\ref{thm:hat_x_t}}\label{subsec:MDE_algorithm}


We now prove the convergence result stated in Theorem \ref{thm:hat_x_t}.

\begin{proof}
In this proof, we first prove that the estimate sequence $\{\hat \theta_i^{(\cdot)} (t)\}$ at each node $\mathcal{N}_{i}$ converges almost surely (a.s.), and the limiting value is given by
\begin{align}
\lim_{t \rightarrow \infty} \hat \theta_i^{+} (t) &= \max\left\{\frac{\sum_{i=1}^n \hat h_i^{+} y_i}{\sum_{i=1}^n \hat h_i^{+} + n\delta},0\right\},~\forall i,\nonumber\\
\lim_{t \rightarrow \infty} \hat \theta_i^{-} (t) &= \min\left\{\frac{\sum_{i=1}^n \hat h_i^{-} y_i}{\sum_{i=1}^n \hat h_i^{-} + n\delta},0\right\},~\forall i.\nonumber
\end{align}
with $\hat h_i^{(\cdot)} \in \{0,1\}$ denoting the limiting value of the convergent sequence $\{\hat h_i^{(\cdot)}(t)\}$. We then use the fact that $\lim_{t \rightarrow \infty} \hat {\bar y}_i(t) = \bar y, \forall i$ \cite{Kar11}, to prove the convergence of $\{\hat \theta_i (t)\}$.

Without loss of generality, we only prove the positive case, i.e., $\{\hat \theta_i^+ (t)\}$. In Lemma \ref{lem:x_y_avg_v_avg_0}, we have proved that $\lim_{t \rightarrow \infty} (\hat \theta_i^+(t)-\max\{\hat \theta_{current}^+(t),0\}) =0$. Thus, we only need to show that $\max\{\hat \theta_{current}^+(t),0\}$ converges. Since $\mathcal{G}^+(t) = (\max\{\hat \theta_{current}^+(t),0\} -\varepsilon ,\max\{\hat \theta_{current}^+(t),0\} +\varepsilon)$, the study on $\max\{\hat \theta_{current}^+(t),0\}$ is equivalent to the study on $\mathcal{G}^+(t)$ in term of convergence. By the smooth moving condition, there is at most one $b_k$ in $\mathcal{G}^+(t)$, $\forall t$. Thus, there are two different moving statuses of $\hat \theta_{current}^+(t)$ at each iteration cataloged by the number of boundaries in $\mathcal{G}^+(t)$:


\begin{itemize}
\item Case 1: No boundaries belong to $\mathcal{G}^+(t)$, i.e., $b_k \not\in \mathcal{G}^+(t)$, $\forall k$. In other words, $\mathcal{G}^+(t)$ belongs to a single interval $\mathcal{I}_j$, i.e., $\mathcal{G}^+(t) \subseteq \mathcal{I}_j$.


\item Case 2: A boundary exists in $\mathcal{G}^+(t)$, i.e., $\exists k, b_k \in \mathcal{G}^+(t)$.

%
%
\end{itemize}

In Appendix~\ref{App_approach}, we provide Lemmas \ref{lem:motion_case_1} through \ref{lem:case_2_con}. Specifically, in Lemmas  \ref{lem:motion_case_1}, \ref{lem:motion_case_2}, and \ref{lem:swithcing_case_1_case_2}, we prove the convergence of $\max\{\hat \theta_{current}^+(t),0\}$. In particular, we show that $\max\{\hat \theta_{current}^+(t),0\}$ either converges or the moving status switches to the other one for Case 1 and Case 2 in Lemma \ref{lem:motion_case_1} and Lemma \ref{lem:motion_case_2}, respectively. For the moving status switching, Lemma \ref{lem:swithcing_case_1_case_2} further shows that the number of switching between Case 1 and Case 2 is finite, which implies the convergence of $\max\{\hat \theta_{current}^+(t),0\}$. In Lemmas \ref{lem:case_1_con} and \ref{lem:case_2_con}, we further derive the limiting values for Case 1 and Case 2, respectively.

Together with the fact that $\lim_{t \rightarrow \infty} \hat {\bar y}_i(t) = \bar y, \forall i$, the convergence of $\hat \theta_i(t)$ is guaranteed, which could be expressed as
\begin{equation}
\lim_{t \rightarrow \infty} \hat \theta_i (t) = \left\{
   \begin{array}{c}
   \max\left\{\frac{\sum_{i=1}^n \hat h_i^{+} y_i}{\sum_{i=1}^n \hat h_i^{+} + n \delta},0\right\},~ {\bar y} \geq 0 \\
   \min\left\{\frac{\sum_{i=1}^n \hat h_i^{-} y_i}{\sum_{i=1}^n \hat h_i^{-} + n \delta},0\right\},~{\bar y} < 0  \\
   \end{array}, \forall i.
  \right.
\end{equation}
\end{proof}

\section{Proofs for Theorem~\ref{thm:h_1>h_2} and Theorem~\ref{thm:E_theta_to_theta}}\label{sec:Performance}
In this section, we derive the expectation and the variance of local estimate with the proposed algorithm. In Section \ref{sec:Proof_Con}, we have proven that $\hat \theta_i^{+}$ in the MDE algorithm converges to $\max\left\{\frac{\sum_i \hat h_i^{+} y_i/n}{\sum_i \hat h_i^{+}/n + \delta},0\right\}$. Since $\delta$ can be arbitrarily small, we approximate the converged value $\hat \theta^{+}$ as $\max\left\{\frac{\sum_i \hat h_i^{+} y_i}{\sum_i \hat h_i^{+}},0\right\}$ here. In addition, the converged values $\hat \theta^{+}$'s (even with the same initial observations) may be different over different network realizations. In particular, the proposed algorithm might lead to random realizations of $\hat \theta^{+}$ and $\mathbf{\hat h}^{+}$, which satisfy
\begin{align}
&(\hat \theta^+)^2 - 2 y_i \hat \theta^+ \overset{\hat h_i^+=0}{\underset{\hat h_i^+=1}\gtrless} 2 \sigma^2 \ln \frac{p_1}{p_0}\label{equ:hard_decision_a_i}\\
&\hat \theta^+ = \max\left\{\frac{\sum_i \hat h_i^{+} y_i}{\sum_i \hat h_i^{+}},0\right\} \geq 0,
\end{align}
where $\mathbf{\hat h}^{+}$ is a random vector denoting $[\hat h_1^{+},\cdots,\hat h_n^{+}]$. In total, there are $2^n$ possible random values for $\mathbf{\hat h}^{+}$. In order to derive a meaningful result, we adopt order statistics into the rest of the analysis. In Subsection \ref{subsec:shrink}, we first prove Theorem \ref{thm:h_1>h_2} to establish the shrinking over the dimension of the probability space from $2^n$ to $2n$, with a more structured format when we order the observations. We then study the expectation of $\hat \theta$ in Theorem \ref{thm:E_theta_to_theta} at Subsection \ref{subsec:E_hat_theta}. We also study the variance $\operatorname{Var}(\hat \theta^{(\cdot)})$ of $\hat \theta$, whose elements are derived respectively in Subsections \ref{subsec:var_theta}, \ref{subsec:sta_fea_Y_Y_k}, and \ref{subsec:Prob_h_k}.


\subsection{Shrinking the Probability Space of $\mathbf{\hat h}^{(\cdot)}$}\label{subsec:shrink}
In this subsection, we prove Theorem \ref{thm:h_1>h_2} to establish the shrinking over the probability space of interest when we order the observations.
\begin{proof}
Here we only prove the $\mathbf{\hat h}^{+}$ part, for the proof of the $\mathbf{\hat h}^{-}$ part is similar.
We define the decision region of $\hat h_{(i)}^+$ as $\mathcal{D}_{(i)}^+$, which is the region of $\hat \theta^+$ when $\hat h_{(i)}^+ =1$.
By (\ref{equ:hard_decision_a_i}), $\mathcal{D}_{(i)}^+$ can be expressed as:

1) If $y_{(i)}^2 + 2 \sigma^2 \ln \frac{p_1}{p_0} < 0$, we have $\hat h_i^+ =0$ for any $\hat \theta^+$. Thus, we have $\mathcal{D}_{(i)}^+ =\emptyset$;

2) If $y_{(i)}^2 + 2 \sigma^2 \ln \frac{p_1}{p_0} \geq 0$, we have $\mathcal{D}_{(i)}^+ = \Big [y_{(i)} - \sqrt{y_{(i)}^2 + 2 \sigma^2 \ln \frac{p_1}{p_0}}, y_{(i)} + \sqrt{y_{(i)}^2 + 2 \sigma^2 \ln \frac{p_1}{p_0}} \Big]$.

The proof here is equivalent to proving that $\mathcal{D}_{(1)}^+ \subseteq \mathcal{D}_{(2)}^+ \subseteq ... \subseteq \mathcal{D}_{(n)}^+$ is true.
Next, we prove the above statement for both of the two cases: $p_1 \geq 0.5$ and $p_1 < 0.5$.

\begin{flushleft}
Case 1: $p_1 \geq 0.5$. In this case, we have $2 \sigma^2 \ln \frac{p_1}{p_0} \geq 0$ and $\mathcal{D}_{(i)}^+ \neq \emptyset$ for all $i$.
For the upper boundaries of $\mathcal{D}_{(i)}^+$'s, they are increasing with their index $i$, which could be proven by showing that $r(y) = y + \sqrt{y^2 + 2 \sigma^2 \ln \frac{p_1}{p_0}}$ is a monotonic increasing function when $2 \sigma^2 \ln\frac{p_1}{p_0} \geq 0$, i.e.,
\begin{align}
r'(y) &=\left(y + \sqrt{y^2 + 2 \sigma^2 \ln \frac{p_1}{p_0}}\right)' \nonumber\\
&= 1 + \frac{y}{\sqrt{y^2 + 2 \sigma^2 \ln \frac{p_1}{p_0}}} > 0. \label{equ:upper_boundaru_increasing}
\end{align}
For the lower boundaries of $\mathcal{D}_{(i)}^+$'s, they are all negative. Since $\hat \theta^+$ is always positive, the negative part of $\mathcal{D}_{(i)}^+$'s are infeasible. Thus, we redefine $\mathcal{D}_{(i)}^+ = \left[0, y_{(i)} + \sqrt{y_{(i)}^2 + 2 \sigma^2 \ln \frac{p_1}{p_0}}\right]$ in this case. Thus, we conclude that $\mathcal{D}_{(1)}^+ \subseteq \mathcal{D}_{(2)}^+ \subseteq ... \subseteq \mathcal{D}_{(n)}^+$ when $p_1 \geq 0.5$.
\end{flushleft}
Case 2: $p_1 < 0.5$. In this case, we have $2 \sigma^2 \ln \frac{p_1}{p_0} < 0$.  Next, we derive the expression of $\mathcal{D}(i)^+$ for different values of $y_{(i)}$.
When $y_{(i)}^2 < - 2 \sigma^2 \ln \frac{p_1}{p_0}$, we have $\mathcal{D}_{(i)}^+ =\emptyset$;
When $y_{(i)}^2 \geq - 2 \sigma^2 \ln \frac{p_1}{p_0}$, for the case of $y_{(i)} \leq -\sqrt{- 2 \sigma^2 \ln \frac{p_1}{p_0}}$, $\mathcal{D}_{(i)}^+ =  \left[y_{(i)} - \sqrt{y_{(i)}^2 + 2 \sigma^2 \ln \frac{p_1}{p_0}}, y_{(i)} + \sqrt{y_{(i)}^2 + 2 \sigma^2 \ln \frac{p_1}{p_0}}\right]$ is in the negative field. Since $\hat \theta^+$ is always positive, the case of $y_{(i)} \leq -\sqrt{- 2 \sigma^2 \ln \frac{p_1}{p_0}}$ is infeasible. Thus we only need to consider the case of $y_{(i)} \geq \sqrt{- 2 \sigma^2 \ln \frac{p_1}{p_0}}$. We then have $\mathcal{D}_{(i)}^+ =  \left[y_{(i)} - \sqrt{y_{(i)}^2 + 2 \sigma^2 \ln \frac{p_1}{p_0}}, y_{(i)} + \sqrt{y_{(i)}^2 + 2 \sigma^2 \ln \frac{p_1}{p_0}}\right]$, where the upper boundary is an increasing sequence over $i$ by the same argument as $(\ref{equ:upper_boundaru_increasing})$ and the lower boundary is a positive decreasing sequence over $i$, which could be proven by showing that $j(y) = y - \sqrt{y^2 + 2 \sigma^2 \ln \frac{p_1}{p_0}}$ is a monotonic decreasing function when $2 \sigma^2 \ln\frac{p_1}{p_0} < 0$, i.e.,
\begin{align}
j'(y) &= \left(y - \sqrt{y^2 + 2 \sigma^2 \ln \frac{p_1}{p_0}}\right)' \nonumber\\
&= 1 - \frac{y}{\sqrt{y^2 + 2 \sigma^2 \ln \frac{p_1}{p_0}}} <0.
\end{align}
Therefore, we have the same conclusion as the previous case, and we conclude that  $\mathcal{D}_{(1)}^+ \subseteq \mathcal{D}_{(2)}^+ \subseteq ... \subseteq \mathcal{D}_{(n)}^+$ as desired.
\end{proof}

We denote the corresponding convergence vector according to the ordered observations as a random vector $\mathbf{h}^{(\cdot)}$. Although there are totally $2^n$ possible values for $\mathbf{\hat h}^{(\cdot)}$, only $n$ possible values are in the probability space of $\mathbf{h}^{+}$ or $\mathbf{h}^{-}$, i.e., $\mathbf{h}_1^+ = [1,1,...,1], \mathbf{h}_2^+ = [0,1,...,1],...,\mathbf{h}_n^+ = [0,0,...,0,1]$, and $\mathbf{h}_1^- = [1,1,...,1], \mathbf{h}_2^- = [1,...,1,0],...,\mathbf{h}_n^- = [1,0,...,0,0]$, which means that the possible values of $\mathbf{h}^+$ could only be in the form that starts with successive $0$'s and followed with successive $1$'s, with similar rules held for $\mathbf{h}^{-}$.


\subsection{Expectation of $\hat \theta$}\label{subsec:E_hat_theta}
In this subsection, we prove Theorem \ref{thm:E_theta_to_theta} to derive the expected value of the achieved estimate.
\begin{proof}
Without loss of generality, for the $n$ given observations of $\theta$, we denote the $k$ invalid observations as $Y_1, Y_2, ..., Y_k$, with $Y_j \sim N(0, \sigma^2)$, $j \in \{1, ..., k\}$,  and the $n-k$ valid observations as $Y_{k+1}, Y_{k+2}, ... , Y_{n}$, with $Y_j \sim N(\theta, \sigma^2)$, $j \in \{k+1, ..., n\}$.

We first prove $\mathop{\rm \operatorname{sgn}}(\hat {\bar y}_i) \overset{\text{p}}{\rightarrow} \mathop{\rm \operatorname{sgn}}(\theta)$ (where $\overset{\text{p}}{\rightarrow}$ denotes convergence in probability), $\forall i$, as $\text{SNR} \rightarrow \infty$, where $\mathop{\rm \operatorname{sgn}}$ is a function such that $\mathop{\rm \operatorname{sgn}}(x)=+$ when $x\geq 0$ and $\mathop{\rm \operatorname{sgn}}(x)=-$ when $x< 0$. Since $\hat {\bar{y}}_i \overset{p}\rightarrow \bar y$, $\forall i$ \cite{Kar11}, it is enough to show that $\mathop{\rm \operatorname{sgn}}(\bar y) \overset{\text{p}}{\rightarrow} \mathop{\rm \operatorname{sgn}}(\theta)$. The mean of $y_i$ could be expressed as,
\begin{eqnarray}
\bar y = \frac{\sum_i y_i}{n} = \frac{k}{n}\theta +\frac{\sum_i w_i}{n}.
\end{eqnarray}
Since $w_i$'s are i.i.d. Gaussian white noises with zero mean and variance $\sigma^2$, $\frac{\sum_i w_i}{n}$ is Gaussian random variable with zero mean and variance $\sigma^2/n$. Thus, the error probability is given as,
\begin{align}
\Pr\{\mathop{\rm \operatorname{sgn}}(\bar y) \neq \mathop{\rm \operatorname{sgn}}(\theta)\} = Q \left (\frac{\frac{k}{n}\theta}{\frac{\sigma}{\sqrt{n}}} \right)
< \frac{1}{2} e^{-\frac{k^2 \theta^2}{2 \sigma^2 n}}.
\end{align}
Thus, $\mathop{\rm \operatorname{sgn}}(\bar y) \overset{\text{p}}{\rightarrow} \mathop{\rm \operatorname{sgn}}(\theta)$, as $\text{SNR} \rightarrow \infty$.

Next, we prove that $\operatorname{E}(\hat \theta^+) \overset{\text{p}}{\rightarrow} \theta$ (for the case of $\hat \theta^-$, the proof is similar and skipped). Define
$
\hat \theta_c = \frac{\sum_{i=k+1}^{n} Y_i}{n-k}.
$
Thus, $\operatorname{E}(\hat \theta_c) = \theta$. Define the probability of successful estimate as,
$P^+_c =\Pr\{\hat \theta^+ = \hat \theta_c\}$.

In the following part, we prove that $P^+_c \rightarrow 1$, as $\text{SNR} \rightarrow \infty$ for both of the two cases: $p_1 \geq 0.5$ and $p_1 <0.5$. When $p \geq 0.5$, $P^+_c$ can be expressed with the boundaries of the decision regions:
\begin{align}
&P^+_c = \Pr\left\{ \max_{j \in \{1, ..., k\} } \left ( Y_j + \sqrt{Y_j^2 + 2 \sigma^2 \ln \frac{p_1}{p_0}}  \right )\right.\nonumber\\
&\left.\leq \frac{\sum_{i=k+1}^{n} Y_i}{n-k}  \leq  \min_{j \in \{k+1, ..., n\} } \left ( Y_j + \sqrt{Y_j^2 + 2 \sigma^2 \ln \frac{p_1}{p_0}} \right ) \right\}.
\end{align}
Thus, the union bound of the probability of error, $P^+_e = 1- P^+_c$, could be expressed as
{\small\begin{align}
&P_e \leq\nonumber\\
 &\Pr\left\{\!\min_{j \in \{k+1, ..., n\} }\!\left ( Y_j + \sqrt{Y_j^2 + 2 \sigma^2 \ln \frac{p_1}{p_0}} \right ) \!\leq\! \frac{\sum_{i=k+1}^{n} Y_i}{n-k} \right\} \nonumber\\
 &+ \Pr\left\{ \frac{\sum_{i=k+1}^{n} Y_i}{n-k}\!\leq\!\max_{j \in \{1, ..., k\} } \left ( Y_j + \sqrt{Y_j^2 + 2 \sigma^2 \ln \frac{p_1}{p_0}} \right )\right\}
\end{align}
}
where both of the above two items go to $0$ as $\text{SNR} \rightarrow \infty$.

The proof for the case of $p_1 < 0.5$ is similar.
Therefore, we conclude that $\lim_{\text{SNR} \rightarrow \infty}\operatorname{E}(\hat \theta^+) = \operatorname{E}(\hat \theta_c) = \theta$. Similarly we could have $\lim_{\text{SNR} \rightarrow \infty} \operatorname{E}(\hat \theta^-) = \operatorname{E}(\hat \theta_c) = \theta$. Together with the result in the first part for $\mathop{\rm \operatorname{sgn}}(\bar y) \overset{\text{p}}{\rightarrow} \mathop{\rm \operatorname{sgn}}(\theta)$, as $\text{SNR} \rightarrow \infty$, we have $\lim_{\text{SNR} \rightarrow \infty} \operatorname{E}(\hat \theta) = \theta$.
\end{proof}

\subsection{Variance of $\hat \theta$}\label{subsec:var_theta}
In this subsection, we derive the variance of $\hat \theta$. We have ordered the observations as $y_{(1)}\leq y_{(2)}\leq ...\leq y_{(n)}$, and we define the corresponding random variables as $Y_{(1)} \leq Y_{(2)} \leq ... \leq Y_{(n)}$.

Conditioned on $\mathbf{h}$, the variance of $\hat \theta$ can be derived as
\begin{eqnarray}
\operatorname{Var}(\hat \theta) = \operatorname{E}(\operatorname{Var}(\hat \theta \mid \mathbf{h}))+\operatorname{Var}(\operatorname{E}(\hat \theta\mid \mathbf{h})). \label{equ:var_hat_X}
\end{eqnarray}

The first term on the right-hand side of (\ref{equ:var_hat_X}) can be expressed as,
\begin{align}\label{eq_var}
\operatorname{E}(\operatorname{Var}(\hat \theta \mid \mathbf{h}))
&= \sum_{k=1}^{n} \operatorname{Var}(\hat \theta^{k+} ) \Pr\{\mathbf{h}= \mathbf{h}_k^+\} \nonumber\\
&+ \sum_{k=1}^{n} \operatorname{Var}(\hat \theta^{k-} ) \Pr\{\mathbf{h}= \mathbf{h}_k^-\},
\end{align}
where $\hat \theta^{k+}$ and $\hat \theta^{k-}$ are the estimates when $\mathbf{h}= \mathbf{h}_k^+$ and $\mathbf{h}= \mathbf{h}_k^-$, respectively, i.e.,
\begin{align}
\hat \theta^{k+} &=\frac{\sum_i \hat h_i^+ Y_i}{\sum_i \hat h_i^+} = \frac{\sum_{i=k}^{n} Y_{(i)}}{n-k+1}, \label{equ:hat_theta_k_+}\\
\hat \theta^{k-} &=\frac{\sum_i \hat h_i^- Y_i}{\sum_i \hat h_i^-} = \frac{\sum_{i=1}^{n-k+1} Y_{(i)}}{n-k+1}, \label{equ:hat_theta_k_-}
\end{align}
and the variances of $\hat \theta^{k+}$ and $\hat \theta^{k-}$ can be expressed as,
\begin{align}
\operatorname{Var}(\hat \theta^{k+}) = \operatorname{Var}\left(\frac{\sum_{i=k}^n Y_{(i)}}{n-k+1}\right) 
= \frac{\sum_{i=k}^{n} \sigma^2_{(i)}}{(n-k+1)^2},\nonumber\\
\operatorname{Var}(\hat \theta^{k-}) = \operatorname{Var}\left(\frac{\sum_{i=1}^{n-k+1} Y_{(i)}}{n-k+1}\right)
= \frac{\sum_{i=1}^{n-k+1} \sigma^2_{(i)}}{(n-k+1)^2},\nonumber
\end{align}
where $\sigma^2_{(i)}$ is the variance of $Y_{(i)}$, which will be derived in the next subsection.

The second term on the right-hand side of (\ref{equ:var_hat_X}) can be expressed as,
\begin{align}
&\operatorname{Var}(\operatorname{E}(\hat \theta\mid \mathbf{h}))\nonumber\\
&= \operatorname{E} \big((\operatorname{E}(\hat \theta\mid \mathbf{h}))^2\big) - \operatorname{E}^2(\operatorname{E}(\hat \theta\mid \mathbf{h})) \nonumber\\
&= \big[ \sum_{k=1}^n \operatorname{E}^2(\hat \theta^{k+} ) \Pr \{\mathbf{h}=\mathbf{h}_k^+\} + \sum_{k=1}^n \operatorname{E}^2(\hat \theta^{k-} ) \Pr \{\mathbf{h}=\mathbf{h}_k^-\}\big] \nonumber\\
&  - \big[\sum_{k=1}^n \operatorname{E}(\hat \theta^{k+} ) \Pr \{\mathbf{h}=\mathbf{h}_k^+\} + \sum_{k=1}^n \operatorname{E}(\hat \theta^{k-} ) \Pr \{\mathbf{h}=\mathbf{h}_k^-\} \big ]^2,
\end{align}
where the expectation of $\hat \theta^{k(\cdot)}$ can be derived as,
\begin{align}
\operatorname{E}(\hat \theta^{k+} ) = \frac{\operatorname{E}(\sum_{i=k}^n Y_{(i)})} {n-k+1}
= \frac{\sum_{i=k}^n \mu_{(i)}}{n-k+1},
\end{align}
with $\mu_{(i)}$ as the mean of $Y_{(i)}$, which will be derived in the next subsection. For the negative part, similarly, we have
\begin{align}
\operatorname{E}(\hat \theta^{k-} ) = \frac{\operatorname{E}(\sum_{i=1}^{n-k+1} Y_{(i)})} {n-k+1}
= \frac{\sum_{i=1}^{n-k+1} \mu_{(i)}}{n-k+1}.
\end{align}

 From the above expressions, we see that both of the two terms on the right-hand side of (\ref{equ:var_hat_X}) are constructed by three basic elements, i.e., $\Pr \{\mathbf{h}=\mathbf{h}_k^{(\cdot)}\}$'s, $\mu_{(i)}$'s, and $\sigma^2_{(i)}$'s. In the following subsections, we derive them by exploring the statistics of $Y_{(i)}$.

\subsection{Statistics of~$Y_{(i)}$}\label{subsec:sta_fea_Y_Y_k}
 First, we start from the pdf of $Y$, where the received signal $Y$ is a random variable, which is the sum of two independent random variables, i.e., $Y=h \theta + W$, where $\Pr(h \theta=\theta)=p_1$ and $\Pr(h \theta=0)=p_0$, and $W$ is an independent Gaussian random variable with zero mean and variance $\sigma^2$. The pdf of $Y$ can be expressed as

\begin{equation}\label{equ:f_Y(y)}
f_Y(y) = \frac{1}{\sigma\sqrt{2\pi}} e^{ -\frac{y^2}{2\sigma^2} } p_0 + \frac{1}{\sigma\sqrt{2\pi}} e^{ -\frac{(y-\theta)^2}{2\sigma^2} } p_1 \nonumber
\end{equation}
%
and its cdf is expressed as
\begin{equation}
F_Y(y)= \Phi\left(\frac{y}{\sigma}\right) \cdot p_0 + \Phi\left(\frac{y-\theta}{\sigma}\right) \cdot p_1,\nonumber
\end{equation}
where $\Phi(x) = \frac{1}{\sqrt{2\pi}} \int_{-\infty}^x e^{-t^2/2} \, dt$.
%


Next, we derive the cdf of the ordered received signals $Y_{j}$'s. The cdf of $Y_{(i)}$ can then be expressed as
\begin{align}
&F_{Y_{(i)}}(r)= \Pr\{ Y_{(i)} < r \} \nonumber\\
&= \Pr \{\text{the~number~of~$Y_j$~less~than~or~equal~to~$r$~is~at~least~$i$} \}\nonumber\\
&= \sum_{k=i}^{n} \binom{n}{k} F_Y^k(r)[1- F_Y(r)]^{n-k}.
\end{align}



The joint pdf of $Y_{(k_1)},Y_{(k_2)},...,Y_{(k_j)}$, ($1 \leq k_1 < k_2 <...<k_j\leq n; 1 \leq j \leq n$), is, for $y_1\leq y_2 \leq \cdots \leq y_j$,
\begin{align}
&  f_{k_1 k_2 \cdots k_j} (y_1,y_2,...,y_j)\nonumber\\
&=\frac{n! \cdot F_Y^{k_1-1}(y_1) f_Y(y_1)[F_Y(y_2)-F_Y(y_1)]^{k_2-k_1-1}f_Y(y_2)} {(k_1 -1)! (k_2-k_1-1)! \cdots (n-k_j)!} \nonumber\\
&~~\times\cdots [1-F_Y(y_k)]^{n-k_j}f_Y(y_j)
\end{align}
By the result in \cite{David03}, the mean of $Y_{(i)}$ can be calculated as
\begin{align}
\mu_{(i)} &= n \binom{n-1}{i-1}\!\int_{-\infty}^{\infty} x [F_Y(x)]^{i-1} [1-F_Y(x)]^{n-i} f_Y(x) d x \nonumber\\
&= n \binom{n-1}{i-1}\!\int_{0}^{1} F_Y^{-1}(u) u^{i-1} (1-u)^{n-i} d u,
\end{align}
and the variance of $Y_{(i)}$ is given as
\begin{eqnarray}
\sigma^2_{(i)} = E((Y_{(i)})^2) - \mu_{(i)}^2.
\end{eqnarray}

\subsection{Probability of $\mathbf{h}=\mathbf{h}_{k}^{(\cdot)}$}\label{subsec:Prob_h_k}
Next, we derive the probability that $\mathbf{h}$ equals $\mathbf{h}_{k}^{+}$. We have
\begin{align}
\Pr\{\mathbf{h}= &  \mathbf{h}_{k}^{+}\} = \Pr\left\{\hat \theta \in \mathcal{D}_{(i)}, i = k, k+1, ... , n; \right. \nonumber\\
&\left.\hat \theta \not\in \mathcal{D}_{(j)}, j = 1,2, ..., k-1 ; \operatorname{sgn}\left(\sum Y_i\right) = + \right\}.
\end{align}
Specifically, when $p_1 \geq 0.5$, we have
\begin{align}\label{equ:Pr_ak_is_true_1}
&\Pr\{\mathbf{h}= \mathbf{h}_{k}^{+}\}= \nonumber\\
&\Pr\left\{ Y_{(k-1)} + \sqrt{Y_{(k-1)}^2 + 2 \sigma^2 \ln\frac{p_1}{p_0}} \leq \frac{\sum_{i=k}^n Y_{(i)}}{n-k+1}\right.\nonumber\\
&\left. \leq Y_{(k)} + \sqrt{Y_{(k)}^2 + 2 \sigma^2 \ln\frac{p_1}{p_0}} ; \operatorname{sgn}\left(\sum Y_i\right) =+\right\}.
\end{align}

When $p_1 <0.5$, we have
\begin{align}\label{equ:Pr_ak_is_true_0}
&\Pr\{\mathbf{h}= \mathbf{h}_{k}^{+}\} \nonumber\\
&= \Pr\left\{ Y_{(k-1)} + \sqrt{Y_{(k-1)}^2 + 2 \sigma^2 \ln\frac{p_1}{p_0}} \leq \frac{\sum_{i=k}^n Y_{(i)}}{n-k+1} \nonumber\right.\\
&~~~~~~~~\leq Y_{(k)} + \sqrt{Y_{(k)}^2 + 2 \sigma^2 \ln\frac{p_1}{p_0}} ; \nonumber\\
&~~~~~~~~ \left. Y_{(k-1)} \geq -\sqrt{-2 \sigma^2 \ln\frac{p_1}{p_0}} ; \operatorname{sgn}\left(\sum Y_i\right) = + \right\} \nonumber\\
&+\Pr\left\{ \frac{\sum_{i=k}^n Y_{(i)}}{n-k+1} \leq Y_{(k)} + \sqrt{Y_{(k)}^2 + 2 \sigma^2 \ln\frac{p_1}{p_0}} ;\right.\nonumber\\
 & ~~~~~~~~Y_{(k)} \geq -\sqrt{-2 \sigma^2 \ln\frac{p_1}{p_0}} ; \nonumber\\
& ~~~~~~~~\left. Y_{(k-1)} < -\sqrt{-2 \sigma^2 \ln\frac{p_1}{p_0}} ; \operatorname{sgn}\left(\sum Y_i\right) = +\right\} \nonumber\\
\end{align}
\begin{align}
&+\Pr\left\{ Y_{(k-1)} - \sqrt{Y_{(k-1)}^2 + 2 \sigma^2 \ln\frac{p_1}{p_0}} \geq \frac{\sum_{i=k}^n Y_{(i)}}{n-k+1} \nonumber \right.\\
& ~~~~~~~~ \geq Y_{(k)} - \sqrt{Y_{(k)}^2 + 2 \sigma^2 \ln\frac{p_1}{p_0}} ;  \nonumber\\
& ~~~~~~~~ \left. Y_{(k-1)} \geq -\sqrt{-2 \sigma^2 \ln\frac{p_1}{p_0}} ; \operatorname{sgn}\left(\sum Y_i\right) = + \right\} \nonumber\\
&+ \Pr\left\{\frac{\sum_{i=k}^n Y_{(i)}}{n-k+1} \geq Y_{(k)} - \sqrt{Y_{(k)}^2 + 2 \sigma^2 \ln\frac{p_1}{p_0}} ; \right. \nonumber\\
& ~~~~~~~~ Y_{(k)} \geq -\sqrt{-2 \sigma^2 \ln\frac{p_1}{p_0}};  \nonumber\\
& ~~~~~~~~\left. Y_{(k-1)} < -\sqrt{-2 \sigma^2 \ln\frac{p_1}{p_0}} ; \operatorname{sgn}\left(\sum Y_i\right) = + \right\}.
\end{align}

The expression for the negative case of $\mathbf{h}_{k}^{-}$ is similar, which is omitted here. So far, all the terms in (\ref{equ:var_hat_X}) have been calculated. Thus, the closed-form variance could be derived. However, this expression is too complicated to make any intuitive observations. In the next section, we analyze the asymptotic performance of the proposed algorithm, which could lead to some compact and intuitive observations.

\section{Asymptotic Analysis}\label{sec:Asym_Analysis}

In the previous section, we studied the mean and variance of the limiting value with the proposed algorithm. In this section, we study the asymptotic performance of the proposed algorithm as $n \rightarrow \infty$. We first review the asymptotic theory of order statistics, then we study the asymptotic result of the given estimator. Afterwards, we show that $\operatorname{Var}(\hat \theta)$ is of the same order as $\operatorname{Var}(\hat \theta_\text{Ideal})$ when $n$ tends to infinity.

In the asymptotic theory of order statistics \cite{David03}, the limiting distributions of appropriately standardized sequences of $k$th order statistics $\{X_{(k)}\}$ as the number of samples $n$ tends infinity are studied. Generally, the order number $k$ can change as a function of $n$.
If $\lim_{n \rightarrow \infty} k/n$ exists between 0 and 1, but not equal to $0$ or $1$, the corresponding order statistics $X_{(k)}$ of the sequence $\{X_{(k)}\}$ are called the central order statistics. Otherwise, they are called the extreme order statistics.

In mathematical statistics, central order statistics are used to construct consistent sequences of estimators for quantiles of the unknown distribution $F(u)$ based on the realization of a random vector $X$. 
For instance, let $x_q$ be a quantile at level $q$, $(0<q<1)$, of the distribution function $F(u)$ with a continuous probability density $f(u)$ and strictly positive in some neighborhood of the point $x_q$. As such, the sequence of central order statistics $\{X_{(k)}\}$ with order numbers $k = \lceil nq \rceil$, where $\lceil \cdot \rceil$ is the ceiling function, is a sequence of consistent estimators for the quantiles $x_q$, as $n \rightarrow \infty$ \cite{David03}.

For a general distribution $F$ with a continuous non-zero density at $F^{-1}(q)$, the $q-$th sample quantile is asymptotically normally distributed as $n$ tends to infinity, and is approximated by
\begin{equation}
\lim_{n \rightarrow \infty} F_{X_{(\lceil nq \rceil)}}(x) = F_{X_{n,q}} (x),\label{equ:X_normal}
\end{equation}
where $X_{n,q}\sim N \left( F^{-1} (q), \frac{q(1-q)}{n[f(F^{-1}(q))]^2} \right)$ \cite{David03}.


In (\ref{equ:hat_theta_k_+}) and (\ref{equ:hat_theta_k_-}), we defined $\hat \theta^{k (\cdot)}$ when $n$ is finite. Next, we derive the limiting value of $\hat \theta^{\lceil nq\rceil(\cdot)}$ when $n \rightarrow \infty$.

\begin{thm}\label{thm:theta_to_int}
If $F_Y$ is a continuous function, for any $0<q<1$ and $\varepsilon>0$, we have
\begin{eqnarray}
\lim_{n \rightarrow \infty} \Pr\left\{\left|\hat \theta^{\lceil nq \rceil +} - \frac {\int_{ F_Y^{-1} (q)}^{+\infty} y f_Y(y) dy}{\int_{ F_Y^{-1} (q)}^{+\infty} f_Y(y) dy} \right | \geq \varepsilon\right\}= 0 , \nonumber\\
\lim_{n \rightarrow \infty} \Pr\left\{\left|\hat \theta^{\lceil nq \rceil -} - \frac {\int^{ F_Y^{-1} (q)}_{-\infty} y f_Y(y) dy}{\int^{ F_Y^{-1} (q)}_{-\infty} f_Y(y) dy} \right| \geq \varepsilon\right\}= 0.\nonumber
\end{eqnarray}
\end{thm}

\begin{proof} Here, we prove the positive part, while the proof of the negative part is similar.
By definition, the cdf of $\hat \theta^{\lceil nq \rceil +}$ can be expressed as
\begin{eqnarray}
F_{\hat \theta^{\lceil nq \rceil+}}(r) = \Pr\left\{\frac{\sum_{i=\lceil nq \rceil}^n Y_{(i)}}{n-\lceil nq \rceil+1} < r \right\}.
\end{eqnarray}
Since $\{Y_{(i)}\}$ is the ordered version of $\{Y_{i}\}$, we have
\begin{eqnarray}
\Pr\left\{  \frac{\sum_{i=\lceil nq \rceil}^n Y_{(i)}}{n-\lceil nq \rceil+1} < r  \right\} = \Pr\left\{ \frac{\sum_{j \in \Omega_{n,q}} Y_{j}}{n- \lceil nq \rceil +1} <r   \right\},
\end{eqnarray}
where $\Omega_{n,q} = \{j: Y_{j} \geq Y_{(\lceil nq \rceil)}, j \in \{1,2,...,n\}\}$.

By (\ref{equ:X_normal}), we have $\lim_{n \rightarrow \infty} F_{Y_{(\lceil nq \rceil)}} (y) = F_{Y_{n,q}} (y)$, where $Y_{n,q} \sim N \left( F^{-1} (q), \frac{q(1-q)}{n[f(F^{-1}(q))]^2} \right)$. Thus, we have
\begin{eqnarray}
\lim_{n \rightarrow \infty} \Pr\{ |Y_{(\lceil nq \rceil)} - F^{-1}(q)|\geq \varepsilon \} = 0.
\end{eqnarray}
Since $Y_j$'s are i.i.d. random variables, we have
\begin{eqnarray}
&\lim_{n \rightarrow \infty} \Pr\left\{ \left|\Pr\left\{ \frac{\sum_{j \in \Omega_{n,q}} Y_{j}}{n- \lceil nq \rceil +1} <r   \right\} \right. \right.\nonumber\\
&\left.\left.- \Pr\left\{ \frac{\sum_{j \in \Omega_{q}} Y_{j}}{n- \lceil nq \rceil +1} <r  \right\} \right| \geq \varepsilon \right\} = 0,\label{equ:sum_Y_to_sum_Y}
\end{eqnarray}
where $\Omega_{q} = \{j: Y_{j} \geq F^{-1} (q), j \in \{1,2,...,n\}\}$.

Also, since we only consider the random variables $Y_j$'s where the index is in $\Omega_{q}$, the cdf of $Y_j$ can be derived from the cdf of $Y$ with a normalization factor $\int_{ F_Y^{-1} (q)}^{+\infty} f_Y(y) dy$ as
\begin{eqnarray}
F_{Y_j; Y_j > F_Y^{-1} (q)}(r) =  \frac{f_Y(r)} {\int_{ F_Y^{-1} (q)}^{+\infty} f_Y(y) dy},~~ r > F_Y^{-1} (q).
\end{eqnarray}
Thus we have
\begin{eqnarray}
\lim_{n \rightarrow \infty} \frac{\sum_{j \in \Omega_{q}} Y_{j}}{n - \lceil nq \rceil + 1} &=& \operatorname{E}(Y_j | j \in \Omega_{q})\nonumber\\
&=& \int_{ F_Y^{-1} (q)}^{+\infty} r \frac { f_Y(r)}{\int_{ F_Y^{-1} (q)}^{+\infty} f_Y(y) dy}dr \nonumber\\
&=& \frac {\int_{ F_Y^{-1} (q)}^\infty y f_Y(y) dy}{\int_{ F_Y^{-1} (q)}^{+\infty} f_Y(y) dy}, \label{equ:sum_r_to_int_f}
\end{eqnarray}
which is a constant.

Combining the results in (\ref{equ:sum_Y_to_sum_Y}) and (\ref{equ:sum_r_to_int_f}), together with the definition of cdf,  we obtain the desired result.
\end{proof}
Next, we prove Theorem \ref{thm:var_theta->Ideal}.
\begin{proof}
Without loss of generality, for the $n$ given observations of a positive $\theta$ (for the case of negative $\theta$, the proof is similar), we denote the $k$ invalid observations as $Y_1, Y_2, ..., Y_k$, with $Y_j \sim N(0, \sigma^2)$, $j \in \{1, ..., k\}$,  and the $n-k$ valid observations as $Y_{k+1}, Y_{k+2}, ... , Y_{n}$, with $Y_j \sim N(\theta, \sigma^2)$, $j \in \{k+1, ..., n\}$.

Conditioned on $\mathbf{h}$, the variance of $\hat \theta$ can be derived as
\begin{eqnarray}
\operatorname{Var}(\hat \theta) = \operatorname{E}(\operatorname{Var}(\hat \theta \mid \mathbf{h}))+\operatorname{Var}(\operatorname{E}(\hat \theta\mid \mathbf{h})). \label{equ:var_hat_X_n}
\end{eqnarray}

According to \eqref{eq_var}, the first term on the right-hand side of (\ref{equ:var_hat_X_n}) can be expressed as,
\begin{align}
\operatorname{E}(\operatorname{Var}(\hat \theta \mid \mathbf{h}))=&\sum_{i=1}^{n} \operatorname{Var}(\hat \theta^{i+} ) \Pr\{\mathbf{h}= \mathbf{h}_i^+\} \nonumber\\
&+ \sum_{i=1}^{n} \operatorname{Var}(\hat \theta^{i-} ) \Pr\{\mathbf{h}= \mathbf{h}_i^-\},\label{equ:e_var_theta_h}
\end{align}
where
\begin{eqnarray}
\operatorname{Var}(\hat \theta^{i+}) = \frac{\sum_{j=i}^{n} \sigma^2_{(j)}}{(n-i+1)^2},\label{var+}\\
\operatorname{Var}(\hat \theta^{i-}) = \frac{\sum_{j=1}^{n-i+1} \sigma^2_{(j)}}{(n-i+1)^2},\label{var-}
\end{eqnarray}
with $\sigma^2_{(j)}$ as the variance of $Y_{(j)}$, which converges to $\frac{\frac{j}{n}(1-\frac{j}{n})}{n[f (F ^{-1}(\frac{j}{n}))]^2}$ when $n$ goes to infinity by (\ref{equ:X_normal}).

 According to \eqref{equ:Pr_ak_is_true_1}, $\Pr\{\mathbf{h}= \mathbf{h}_i^+\}$ is exponentially decreasing over $\mbox{SNR}$ when $i \neq k$, due to the Gaussian assumption. Similarly, we also have that $\Pr\{\mathbf{h}= \mathbf{h}_i^-\}$ is exponentially decreasing over $\mbox{SNR}$. By combining \eqref{var+} and \eqref{var-} with \eqref{equ:X_normal}, we have that the linear rate of $\operatorname{Var}(\hat \theta^{i (\cdot)})$ changing over $\mbox{SNR}$ is lower than the exponential rate of $\Pr\{\mathbf{h}= \mathbf{h}_i^{(\cdot)}\}$ decreasing over $\mbox{SNR}$ when $\mathbf{h} \neq \mathbf{h}_k^+$. Thus, only the terms with $\Pr\{\mathbf{h}= \mathbf{h}_k^+\}$ are left in (\ref{equ:e_var_theta_h}) as $\mbox{SNR} \rightarrow \infty$ and we have $\operatorname{E}(\operatorname{Var}(\hat \theta \mid \mathbf{h})) \rightarrow \operatorname{E}(\operatorname{Var}(\hat \theta_{\mbox{\scriptsize Ideal}} \mid \mathbf{h}))$ almost surely by the definition of $\hat \theta_{\mbox{\scriptsize Ideal}}$.

The second term on the right-hand side of (\ref{equ:var_hat_X_n}) can be expressed as
\begin{align}
\operatorname{Var}(\operatorname{E}(\hat \theta\mid \mathbf{h}))&= \operatorname{E} \big((\operatorname{E}(\hat \theta\mid \mathbf{h}))^2\big) - \operatorname{E}^2(\operatorname{E}(\hat \theta\mid \mathbf{h})) \nonumber\\
&=  \sum_{i=1}^n \operatorname{E}^2(\hat \theta^{i+} \mid \mathbf{h} = \mathbf{h}_i^+) \Pr \{\mathbf{h}=\mathbf{h}_i^+\} \nonumber\\
&~~~+ \sum_{i=1}^n \operatorname{E}^2(\hat \theta^{i-} \mid \mathbf{h} = \mathbf{h}_i^-) \Pr \{\mathbf{h}=\mathbf{h}_i^-\} \nonumber\\ &  - \left[\sum_{i=1}^n \operatorname{E}(\hat \theta^{i+} \mid \mathbf{h} = \mathbf{h}_i^+) \Pr \{\mathbf{h}=\mathbf{h}_i^+\} \right. \nonumber\\
&~~~\left.+ \sum_{i=1}^n \operatorname{E}(\hat \theta^{i-} \mid \mathbf{h} = \mathbf{h}_i^-) \Pr \{\mathbf{h}=\mathbf{h}_i^-\} \right]^2\label{equ:var_e_theta_h}
\end{align}
where
\begin{align}
\operatorname{E}(\hat \theta^{i+}\mid \mathbf{h}=\mathbf{h}_i^+ ) &=\frac{\sum_{j=i}^n \mu_{(j)}}{n-i+1},\label{mean+}\\
\operatorname{E}(\hat \theta^{i-} \mid \mathbf{h}=\mathbf{h}_i^-) &=\frac{\sum_{j=1}^{n-i+1} \mu_{(j)}}{n-i+1}.\label{mean-}
\end{align}
with $\mu_{(j)}$ as the mean of $Y_{(j)}$, which converges to $F^{-1}\left(\frac{j}{n}\right)$ when $n$ goes to infinity by (\ref{equ:X_normal}).

According to \eqref{equ:Pr_ak_is_true_1}, $\Pr\{\mathbf{h}= \mathbf{h}_i^{(\cdot)}\}$ is exponentially decreasing over $\mbox{SNR}$ when $\mathbf{h} \neq \mathbf{h}_k^+$, due to the Gaussian assumption. By combining \eqref{mean+} and \eqref{mean-} with \eqref{equ:X_normal}, we have that the linear rate of $\operatorname{E}(\hat \theta^{i(\cdot)} \mid \mathbf{h} = \mathbf{h}_i^{(\cdot)})$ changing over $\mbox{SNR}$ is lower than the exponential rate of $\Pr\{\mathbf{h}= \mathbf{h}_i^{(\cdot)}\}$ decreasing over $\mbox{SNR}$ when $\mathbf{h} \neq \mathbf{h}_k^+$. Thus, only the terms with $\Pr\{\mathbf{h}= \mathbf{h}_k^+\}$ are left in (\ref{equ:var_e_theta_h}) as $\mbox{SNR} \rightarrow \infty$ and we have $\operatorname{Var}(\operatorname{E}(\hat \theta\mid \mathbf{h})) \rightarrow \operatorname{Var}(\operatorname{E}(\hat \theta_{\mbox{\scriptsize Ideal}}\mid \mathbf{h}))$ almost surely by the definition of $\hat \theta_{\mbox{\scriptsize Ideal}}$.

Combining the results in the above two parts, we have $\operatorname{Var}(\hat \theta) \rightarrow \operatorname{Var}(\hat \theta_{\mbox{\scriptsize Ideal}})$ almost surely.
\end{proof}

\section{Simulation Results}\label{sec:Simulation_and_Discussion}
\begin{figure}
\vspace{-10pt}
\begin{center}
  \includegraphics[width=8cm]{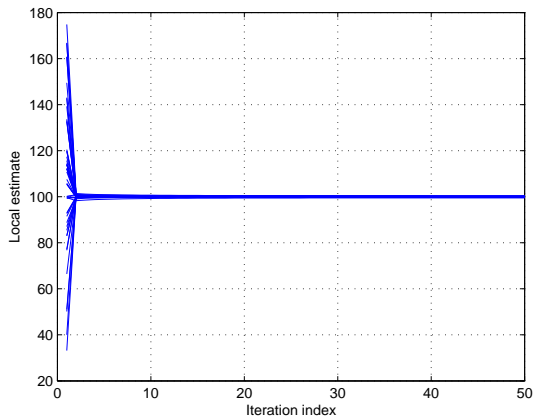}\\
  \vspace{-10pt}
  \caption{The convergence of the MDE algorithm, $\theta=100$.}
  \label{fig:Convergence}
\end{center}
\vspace{-15pt}
\end{figure}

\begin{figure}
\vspace{-2pt}
\begin{center}
  \includegraphics[width=8cm]{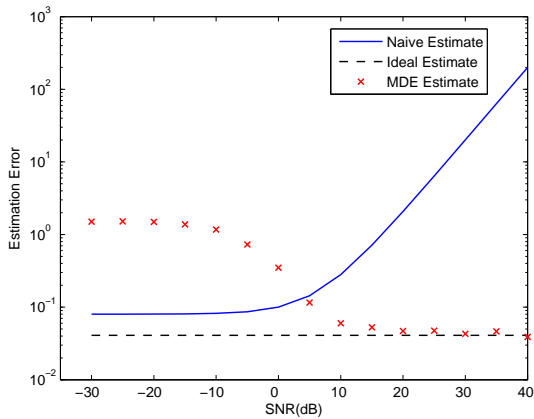}\\
  \vspace{-10pt}
  \caption{The performance comparison among the MDE algorithm, the naive averaging algorithm, and the ideal estimate.}
  \label{fig:Algorithm_Comparison}
\end{center}
\vspace{-15pt}
\end{figure}
In this section, we present simulation results that demonstrate the estimation performance of the proposed MDE algorithm. In our network setting, 50 nodes are uniformly distributed over a unit square where two nodes are connected by an edge if their distance is less than 0.3, which is the predefined transmission range. In addition, $h_i$'s are independently generated with $p_1=0.5$, $w_i$'s are independent white Gaussian noises with zero mean and unit variance, and the other parameter values are specified in the description of each figure.

In Fig. \ref{fig:Convergence}, we demonstrate the convergence (Theorem~\ref{thm:hat_x_t}) of the proposed algorithm. Realizations of the local estimates at the 50 nodes over 50 rounds of iterations, i.e., $\hat \theta_i(t), i \in [1,\cdots,50],t \in [1,\cdots,50]$, are plotted. The target $\theta$ is $100$, which implies $\text{SNR}=40$dB. In the figure, about half of the nodes start around the value 100 and the rest start around 0, indicating that the former ones correspond to valid observations and the latter ones are the nodes with invalid observations. We observe that the local estimates of both types of nodes converge as the number of iteration increases.

In Fig. \ref{fig:Algorithm_Comparison}, we compare the performance of the proposed MDE algorithm with the naive averaging algorithm \eqref{equ:x_naive} and the ideal algorithm \eqref{equ:theta_ideal} discussed in Section~\ref{sec:Network_Model}. In the figure, the estimation error of these three estimates are plotted with SNR ranging from -30 dB to 40 dB. For each SNR, we generate 500 runs of the MDE algorithm, with the limiting consensus value of the local estimate for each realization being taken to be the estimate in the first node at the end of the 3000-th iteration. The estimation error plotted in the figure is the average squared deviation of the limiting consensus value from the true value of $\theta$ over these 500 realizations, i.e., $(\sum (\hat \theta_1(3000)- \theta)^2)/500$. The topology of the communication graph (given by the random node placement) and the observation values across the nodes are independently generated for each realization. We make several observations from this figure. First, the numerical result of the naive averaging algorithm \eqref{equ:x_naive} matches the theoretical results as derived in (\ref{equ:var_x_naive}), i.e., the estimation error variance grows exponentially over SNR; second, the numerical result of the ideal algorithm \eqref{equ:theta_ideal} matches the theoretical results as derived in (\ref{equ:var_ideal}), where the estimation error is the lowest among the three algorithms; and third, although the estimation error of MDE is higher than that of the naive averaging in the lower SNR regime (SNR$<$5dB), it performs much better in the mid and high SNR regimes (SNR$>$20dB), where it approaches the performance of the ideal estimator.

In the following we provide some intuitive explanation of the observed simulation behavior:
1) In the low SNR regime, the target value is relatively small as compared with the Gaussian noise, which leads to a high detection error in \eqref{equ:detection1} and \eqref{equ:detection}. Some invalid observations are wrongly detected as valid ones and negatively incorporated into the estimate update process, whereas, some valid observations are discarded as invalid ones. Thus, the estimate is largely distorted from the ideal estimate, which leads to the poor estimation performance;
2) in the high SNR regime, the detection error in \eqref{equ:detection1} and (\ref{equ:detection}) is very small and almost every observation is correctly detected as valid or invalid. Therefore, the MDE estimate is quite close to the ideal estimate and the MSE of the MDE algorithm approaches the lower bound (i.e., that achieved by the ideal algorithm).
\vspace{-7pt}
\section{Conclusions}\label{sec:Conclusion}
We studied an algorithm named MDE, for distributed estimation of a scalar target signal with imperfect sensing mode information (due to node defects) in a sensor network. For the proposed algorithm, an online learning step assesses the validity of the local observations at each iteration, and then refines the ongoing estimation update process in an iterative fashion. We analytically established the convergence of the MDE algorithm. From the asymptotic results of the performance analysis, we have shown that in the high SNR regime, as the number of nodes goes to infinity, the MDE estimation error converges to that of an ideal estimator with perfect information about the node sensing modes. 



\appendices

\section{Variance of Ideal Estimator in \eqref{equ:var_ideal}} \label{Var_ideal}


We have the first entry of the conditional variance calculated as
\begin{align}
&\operatorname{E} (\operatorname{Var}(\hat \theta_\text{Ideal}^n|\textbf{h})) =\sum_{\textbf{h}} \operatorname{Var}(\hat \theta_\text{Ideal}^n|\textbf{h}) p(\textbf{h})\nonumber\\
                           &= \sum_{k=0}^n \operatorname{Var}\left(\frac{\sum_{i:h_i=1} y_i}{k}\right) \Pr\left\{\sum h_i =k\right\}\nonumber\\
                           &= \sum_{k=0}^n \operatorname{Var}\left(\theta +\frac{\sum_{i:h_i=1} w_i}{k}\right) \Pr\left\{\sum h_i =k\right\}\nonumber\\
                           &= \sigma^2 \cdot \sum_{k=0}^n \frac{1}{k} \binom{n}{k} p_1^k p_0^{n-k}.\nonumber
\end{align}
We have the second entry calculated as
\begin{equation}
\operatorname{Var}(\operatorname{E}(\hat \theta_\text{Ideal}^n|\textbf{h}))= \operatorname{Var}_{\textbf{h}}(\theta)= 0,\nonumber
 \end{equation}
where $\textbf{h}$ is given, and $\operatorname{E}(\hat \theta_\text{Ideal}^n | \textbf{h})= \operatorname{E}\left(\frac{\sum h_i y_i}{\sum h_i} | \textbf{h}\right) = \theta$, which is a constant independent with $\textbf{h}$. Thus we have derived the variance of ideal estimator shown in \eqref{equ:var_ideal}.

%
\section{Lemmas Used in Section~\ref{sec:Proof_Con} to Prove Theorem~\ref{thm:hat_x_t}} \label{App_approach}

\begin{lem}\label{lem:motion_case_1}
If the moving status of $\max\{\hat \theta_{current}^+(t),0\}$ is in Case 1 with $t= t_1$, then $\max\{\hat \theta_{current}^+(t),0\}$ either converges without leaving Case 1 for all $t>t_1$, or the moving status of $\max\{\hat \theta_{current}^+(t),0\}$ changes to Case 2 after $\tilde{t}_1 >t_1$.
\end{lem}
\begin{proof}
If there is a time $\tilde{t}_1$, $\tilde{t}_1 > t_1$, such that the moving status of $\max\{\hat \theta_{current}^+(t),0\}$ changes to Case 2, we have the desired result. Otherwise, for all $t>t_1$, we have that the moving status of $\max\{\hat \theta_{current}^+(t),0\}$ stays in Case 1. In order to show the convergence, we only need to show that $\hat \theta_{current}^+(t)$ is a monotonic and bounded sequence. In this proof, we first prove that $\hat \theta_{goal}^+(t)$ converges when $t>t_1$. After that, we show the monotonicity of $\hat \theta_{current}^+(t)$. At last, we prove that $\hat \theta_{current}^+(t)$ is a bounded sequence for $t>t_1$.

Here we first prove that $\hat \theta_{goal}^+(t)$ converges when $t>t_1$. Since the moving status of $\max\{\hat \theta_{current}^+(t),0\}$ stays in Case 1 for all $t>t_1$, and $\mathcal{G}^+(t)$ cannot jump to a different interval without touching any boundary by the smooth moving condition, we have that $\mathcal{G}^+(t)$ belongs to $\mathcal{I}_j$ for all $t>t_1$. By the definition of $\mathcal{G}^+(t)$, we have $\hat \theta_i^+ (t) \in \mathcal{G}^+(t)$, $\forall i$. Since $\mathcal{G}^+(t)$ belongs to the same $\mathcal{I}_j$ for all $t>t_1$, the inclusion relationship of $\mathcal{G}^+(t)$ and $\mathcal{D}_i$'s do not change for all $t>t_1$. In other words, the detection results of $\hat h_i^+(t)$'s stay the same for all $t>t_1$. Specifically, if we replace $\hat \theta_i^+ (t)$ with an arbitrary $x_j$, $\forall x_j \in \mathcal{I}_j$, the detection result of $\hat h_i^+(t)$ does not change in the detection step (step 2) for any $i$, i.e.,
\begin{eqnarray}\label{equ:x_j_detection_replacement}
x_j^2 - 2 y_i x_j \overset{\hat h_i^+(t)=0}{\underset{\hat h_i^+(t)=1}\gtrless} 2 \sigma^2 \ln \frac{p_1}{p_0}, \forall i, t>t_1.
\end{eqnarray}
 Therefore, we have that $\hat h_i(t)$'s converge. Thus, we conclude that $\hat \theta_{goal}(t)$ converges by the definition. Meanwhile, in the proof, $\hat h_i^+(t)$'s and $\hat \theta_{goal}^+(t)$ are only related to the index of $\mathcal{I}_j$ covering $\mathcal{G}^+(t)$. Thus, we define that $\hat h_i^+(t) = \hat h_i^+[j]$ and $\hat \theta_{goal}^+(t)=\hat \theta_{goal}^+[j]$ by using $j$, the index of $\mathcal{I}_j$.

Next, we show the monotonicity of $\hat \theta_{current}^+(t)$. To this end, we want to prove that
\begin{align}
&\big(\hat \theta_{goal}^+[j] - \hat \theta_{current}^+(t+1)\big) \nonumber\\
&\big(\hat \theta_{current}^+(t+1) - \hat \theta_{current}^+(t)\big) > 0, t>t_1. \label{equ:theta_mono}
\end{align}

By taking average on both sides of (\ref{equ:u_i_t_+}), we have
\begin{eqnarray}
\bar u^+(t) &=& \bar u^+(t-1) + \alpha(t) \big( \{y_i \hat h_i^{+}(t)\}_{avg} - \bar u^{+}(t-1) \big)\nonumber\\
&=& (1-\alpha(t)) \bar u^+(t-1) + \alpha(t) \{y_i \hat h_i^{+}(t)\}_{avg}.
\end{eqnarray}
Similarly, by taking average on both sides of (\ref{equ:v_i_t_+}), we have
\begin{eqnarray}
\bar v^+(t) &=& (1-\alpha(t)) \bar v^+(t-1) + \alpha(t) \{ \hat h_i^{+}(t)\}_{avg},
\end{eqnarray}
which is a positive sequence.

Thus, we have
\begin{align}
 &\big(\hat \theta_{goal}^+[j] - \hat \theta_{current}^+(t+1)\big) \big(\hat \theta_{current}^+(t+1) - \hat \theta_{current}^+(t)\big) \nonumber\\
&= \left( \frac{\{y_i \hat h_i^+[j]\}_{avg}}{\{\hat h_i^+[j]\}_{avg}+\delta} - \frac{\bar u^+(t)}{\bar v^+(t) + \delta} \right)\nonumber\\
&~~~~ \left( \frac{\bar u^+(t)}{\bar v^+(t) + \delta} - \frac{\bar u^+(t-1)}{\bar v^+(t-1) + \delta} \right) \nonumber\\
&= \frac{\alpha(t)(1-\alpha(t))\Upsilon^2}{(\bar v^+(t) + \delta)((1-\alpha(t)) \bar v^+(t-1) + \alpha(t) \{ \hat h_i^{+}[j]\}_{avg} + \delta)^2} \nonumber\\
&~~~~\frac{1}{(\{\hat h_i^+[j]\}_{avg}+\delta)}\label{equ:alpha_t_1-_alpha_t}
\end{align}
where $\Upsilon = \{y_i \hat h_i^+[j]\}_{avg}(\bar v^+(t-1) + \delta) - \bar u^+(t-1) (\{ \hat h_i^{+}[j]\}_{avg} + \delta)$. Note that all of the elements multiplied together in (\ref{equ:alpha_t_1-_alpha_t}) are positive. 

At last, we prove that $\hat \theta_{current}^+(t)$ is a bounded sequence for $t>t_1$. Since both $\bar u^+(t)$ and $\bar v^+(t)$ are bounded by Lemma \ref{lem:y_avg_upper_bounded} and both $\bar v^+(t)$ and $\delta$ are positive, we conclude that $\hat \theta_{current}^+(t) = \frac{\bar u^+(t)}{\bar v^+(t) + \delta}$ is a bounded sequence.
\end{proof}

\begin{lem}\label{lem:motion_case_2}
If the moving status of $\max\{\hat \theta_{current}^+(t),0\}$ is in Case 2 when $t= t_2$, $\max\{\hat \theta_{current}^+(t),0\}$ either converges without leaving Case 2 for all $t>t_2$, or $\exists \tilde{t}_2,~\tilde{t}_2 >t_2$, such that the moving status of $\max\{\hat \theta_{current}^+(t),0\}$ changes to Case 1 from $\tilde{t}_2$.
\end{lem}
\begin{proof}
If there is a time $\tilde{t}_2$, $\tilde{t}_2 > t_2$, such that the moving status of $\max\{\hat \theta_{current}^+(t),0\}$ changes to Case 1, we have the desired result. Otherwise, for all $t>t_2$, we have that the moving status of $\max\{\hat \theta_{current}^+(t),0\}$ stays in Case 2. Since $\max\{\hat \theta_{current}^+(t),0\} \in \mathcal{G}^+(t)$, $b_k \in \mathcal{G}^+(t)$, and $|\mathcal{G}^+(t)| = 2\varepsilon$, we have $|\hat \theta_{current}^+(t) - b_k| \leq 2 \varepsilon$. Together with the fact that $\varepsilon$ can be arbitrarily small, we conclude with convergence automatically.
\end{proof}

\begin{lem}\label{lem:swithcing_case_1_case_2}
For the moving status of $\max\{\hat \theta_{current}^+(t),0\}$, the number of switching times between Case 1 and Case 2 is finite.
\end{lem}
\begin{proof}
 First, we prove that after coming back to Case 1 from Case 2, the monotonicity of $\hat \theta_{current}^+(t)$ stays the same as the one in the previous Case 1 (i.e., Case 1 before going into Case 2). Then, we prove that the sequence of $\{\hat \theta_{current}^+(t_s)\}$, which is a subsequence of $\{\hat \theta_{current}^+(t)\}$, is also monotonic, where we only consider $\{t_s\}$ at which the moving status is in Case 1. Together with the fact that the number of $b_k$'s is finite, lastly we conclude that the number of switching times between Case 1 and Case 2 is finite.

 By following the above logic flow, we first prove that after coming back to Case 1 from Case 2, the monotonicity of $\hat \theta_{current}^+(t)$ stays the same as the one in the previous Case 1 (i.e., Case 1 before going into Case 2). To this end, since we have proven that $\hat \theta_{current}^+(t)$ changes monotonically in Case 1 by Lemma \ref{lem:motion_case_1}, it is sufficient to show that
\begin{eqnarray}
(\hat \theta_{goal}^+(\acute{t}) - \hat \theta_{current}^+(\acute{t}))(\hat \theta_{goal}^+(\grave{t}) - \hat \theta_{current}^+(\grave{t}) \geq 0, \label{equ:keep_moving_2}
\end{eqnarray}
where $\acute{t}$ is the time before going into Case 2 and $\grave{t}$ is the time after coming out from Case 2. Assume that $b_k$ under concern is one of the boundaries of node $j$, i.e., $b_k \in \{r_j^-, r_j^+\}$. Without loss of generality, we assume that $b_k$ is $r_j^-$ and comes into the gathering region from the right side. Therefore, we have $\hat \theta_{goal}^+(\acute{t}) = \hat \theta_{goal}^+[k]$, $\hat \theta_{goal}^+(\grave{t}) = \hat \theta_{goal}^+[k+1]$, and $\hat \theta_{goal}^+(\acute{t}) > \hat \theta_{current}^+(\acute{t})$. In order to prove (\ref{equ:keep_moving_2}), we only need to show that $\hat \theta_{goal}^+[k+1] > \hat \theta_{current}^+(\grave{t})$. For $\hat \theta_{goal}^+ (\grave{t}-1)$, there are two possible values, i.e., $\hat \theta_{goal}^+[k]$ and $\hat \theta_{goal}^+[k+1]$. Specifically, if $\hat \theta_j (\grave{t})$ is on the right of $b_k$, we have $\hat \theta_{goal}^+(\grave{t}-1) = \hat \theta_{goal}^+[k+1]$; otherwise, we have $\hat \theta_{goal}^+(\grave{t}-1) = \hat \theta_{goal}^+[k]$. Next, we prove $\hat \theta_{goal}^+[k+1] > \hat \theta_{current}^+(\grave{t})$ for both cases.

\begin{flushleft}
1) When $\hat \theta_{goal}^+(\grave{t}-1) = \hat \theta_{goal}^+[k+1]$: By a similar derivation of (\ref{equ:theta_mono}), we have
\begin{align}
&\big(\hat \theta_{goal}^+(\grave{t}-1) - \hat \theta_{current}^+(\grave{t})\big) \nonumber\\
&\big(\hat \theta_{current}^+(\grave{t}) - \hat \theta_{current}^+(\grave{t}-1)\big) > 0. \label{equ:theta_grave_t}
\end{align}
Since the second term on the left-hand side of (\ref{equ:theta_grave_t}) is positive by assumption, we have the first term on the left-hand side of (\ref{equ:theta_grave_t}) is also positive. Thus, we have $\hat \theta_{goal}^+[k+1] > \hat \theta_{current}^+(\grave{t})$ as desired.
\end{flushleft}

\begin{flushleft}
2) When $\hat \theta_{goal}^+(\grave{t}-1) = \hat \theta_{goal}^+[k]$: By definition, $\hat \theta_{goal}^+[k+1]$ can be expressed as
\end{flushleft}
\begin{align}
&\hat \theta_{goal}^+[k+1] = \frac{\sum_i \hat h_i^+[k+1] y_i}{\sum_i \hat h_i^+[k+1] + n \delta} \nonumber\\
&= \frac{\sum_i \hat h_i^+[k] y_i + y_j}{\sum_i \hat h_i^+[k] +1 + n \delta}\nonumber\\
&= \frac{\sum_i \hat h_i^+[k] + n \delta}{\sum_i \hat h_i^+[k] +1 + n \delta}\frac{\sum_i \hat h_i^+[k] y_i}{\sum_i \hat h_i^+[k] + n \delta} \nonumber\\
&~~~~+ \frac{y_j}{\sum_i \hat h_i^+[k] +1 + n \delta} \nonumber\\
&= \frac{\sum_i \hat h_i^+[k] + n \delta}{\sum_i \hat h_i^+[k] +1 + n \delta}\hat \theta_{goal}^+[k] + \frac{y_j}{\sum_i \hat h_i^+[k] +1 + n \delta},
\end{align}
where the sum of the weights on $\hat \theta_{goal}^+[k]$ and $y_j$ equals to 1, i.e., $\frac{\sum_i \hat h_i^+[k] + n \delta}{\sum_i \hat h_i^+[k] +1 + n \delta} + \frac{1}{\sum_i \hat h_i^+[k] +1 + n \delta} = 1$. In order to prove $\hat \theta_{goal}^+[k+1] > \hat \theta_{current}^+(\grave{t})$, we only need to show $\hat \theta_{goal}^+[k] > \hat \theta_{current}^+(\grave{t})$ and $y_j> \hat \theta_{current}^+(\grave{t})$. The first part can be proved by the result in (\ref{equ:theta_grave_t}) and $\hat \theta_{goal}^+(\grave{t}-1) = \hat \theta_{goal}^+[k]$. The second part is due to the smooth moving condition defined in Section~\ref{subsec:smooth_moving_condition}, which implies that $y_j > b_k + 3\varepsilon$ and $\hat \theta_{current}^+(\grave{t}) \leq b_k+2\varepsilon$.

Then, we prove that the sequence of $\{\hat \theta_{current}^+(t_s)\}$, which is a subsequence of $\{\hat \theta_{current}^+(t)\}$, is also monotonic, where we only consider $\{t_s\}$ at which the moving status is in Case 1. Specifically, we need to prove
\begin{eqnarray}
(\hat \theta_{goal}^+(\acute{t}) - \hat \theta_{current}^+(\acute{t}))(\hat \theta_{current}^+(\grave{t}) - \hat \theta_{current}^+(\acute{t}) \geq 0. \label{equ:connection}
\end{eqnarray}

Assume that $b_k$ is the boundary under concern in this visit of Case 2. Without loss of generality, we assume that $b_k$ comes into the gathering region from the right side. Therefore, we have $\hat \theta_{goal}^+(\acute{t}) = \hat \theta_{goal}^+[k]$, $\hat \theta_{goal}^+(\grave{t}) = \hat \theta_{goal}^+[k+1]$, and $\hat \theta_{goal}^+(\acute{t}) > \hat \theta_{current}^+(\acute{t})$. In order to prove (\ref{equ:connection}), we only need to show that $\hat \theta_{current}^+(\grave{t}) > \hat \theta_{current}^+(\acute{t})$. Since the moving status of $\hat \theta_{current}^+(\acute{t}+1)$ and $\hat \theta_{current}^+(\grave{t}-1)$ is in Case 2, we have that both $\hat \theta_{current}^+(\acute{t}+1)$ and $\hat \theta_{current}^+(\grave{t}-1)$ are in the region of $[b_k - 2 \varepsilon, b_k + 2 \varepsilon]$. Together with the assumption that $b_k$ comes into the gathering region from the right side, we have that
$\hat \theta_{current}^+(\acute{t}) \leq b_k - 2\varepsilon$ and $\hat \theta_{current}^+(\grave{t}) \geq b_k + 2 \varepsilon$. Hence, we have $\hat \theta_{current}^+(\grave{t}) > \hat \theta_{current}^+(\acute{t})$ as desired.

So far, we have proved that the overall monotonicity of $\hat \theta_{current}^+(t)$ stays the same as when we only consider the iteration in Case 1, which means that the $b_j$'s for each visit of Case 2 are different. Together with the fact that the number of $b_j$'s is finite, we have that the number of switching from Case 1 to Case 2 is finite as desired.
\end{proof}

\begin{lem}\label{lem:case_1_con}
If the moving status of $\max\{\hat \theta_{current} (t),0\}$ is in Case 1 when $t>t_1$ and $\max\{\hat \theta_{current} (t),0\}$ converges without leaving Case 1 for all $t>t_1$, the limiting value is given by
\begin{eqnarray}
\max\left\{\frac{\sum_{i=1}^{n} \hat h_i^+[j] y_i} {\sum_{i=1}^{n} \hat h_i^+[j] + \delta n},0\right\},
\end{eqnarray}
where $j$ is the index of $\mathcal{I}_j$, $\mathcal{G}^+(t_1) \subseteq \mathcal{I}_j$.
\end{lem}
\begin{proof}
Since $\hat \theta_{current}^+ (t) = \frac{\bar u^+(t)}{\bar v^+(t)+\delta}$, we only need to show that $\lim_{t \rightarrow \infty} \bar u^+(t) = \{\hat h_i^+[j] y_i\}_{avg}$ and $\lim_{t \rightarrow \infty} \bar v^+(t) = \{\hat h_i^+[j]\}_{avg}$. Here we only prove the part of $\bar u^+(t)$, while the proof for the part of $\bar v^+(t)$ is similar. By (\ref{equ:bar_u_t+1}), we have
\begin{align}
&\bar u^{+}(t+1)=\nonumber\\
&\big(1-\alpha(t+1)\big) \bar u^{+}(t) + \alpha(t+1) \{y_i \hat h_i^{+}(t+1)\}_{avg}. \label{equ:bar_u_+_t+1}
\end{align}

\begin{flushleft}
Since the moving status of $\hat \theta_{current}^+ (t)$ converges in Case 1 for all $t>t_1$, we have $\hat h_i^+ (t) = \hat h_i^+ [j]$, $\forall t>t_1$, with a similar derivation as (\ref{equ:x_j_detection_replacement}). Thus, $\{\hat h_i^+(t) y_i\}_{avg}$ is a deterministic value when $t>t_1$ and equals $\{\hat h_i^+[j] y_i\}_{avg}$, $\forall t>t_1$.
\end{flushleft}
Thus, we rewrite (\ref{equ:bar_u_+_t+1}) as
\begin{eqnarray}
\bar u^+(t+1)- K = [1-\alpha(t+1)] [\bar u^+(t) - K], ~t > t_1,
\end{eqnarray}
where $K=\{\hat h_i^+[j] y_i\}_{avg}$. The limiting value of $\bar u^+(t) - K$ can be expressed as
\begin{eqnarray}
\lim_{t \rightarrow \infty} |\bar u^+(t) - K|&= \left|\prod_{t = t_1}^{\infty} (1 - \alpha(t)) [\bar u^+(t_1)-K]\right| \nonumber\\
&\leq \exp^{-\sum_{t = t_1}^{\infty} \alpha(t)} |\bar u^+(t_1)-K|.\label{equ:lim_u_t-K}
\end{eqnarray}
Since $\sum_{t} \alpha(t) = \infty$ and $\alpha(t) \in (0,1)$, we have that the right-hand side of (\ref{equ:lim_u_t-K}) equals 0. Then we conclude that $u(t)$ converges to $K$.
\end{proof}

\begin{lem}\label{lem:case_2_con}
If the moving status of $\max\{\hat \theta_{current} (t),0\}$ is in Case 2 and $\max\{\hat \theta_{current} (t),0\}$ converges without leaving Case 2 for all $t>t_2$ (according to the definition of Case 2, for certain node $c$, $h_c^+$ changes in Case 2), the limiting value is given by either
\begin{eqnarray}\label{equ:case21}
 \max\left\{\frac{\sum_{i=1}^{n} \hat h_i^+[j] y_i} {\sum_{i=1}^{n} \hat h_i^+[j] + \delta n},0\right\},
\end{eqnarray}
when $h_c^+[j]=1$ and $h_c^+[j+1]=0$,
or
\begin{eqnarray}
\max\left\{\frac{\sum_{i=1}^{n} \hat h_i^+[j+1] y_i} {\sum_{i=1}^{n} \hat h_i^+[j+1] + \delta n},0\right\},
\end{eqnarray}
when $h_c^+[j]=0$ and $h_c^+[j+1]=1$,
where $j$ is the index of $\mathcal{I}_j$, $\mathcal{G}^+(t_2) \subseteq \mathcal{I}_j$.
\end{lem}

\begin{proof}
Since the region of $\hat \theta_{current}^+ (t_2)$ in Case 2 is $[b_k-2\varepsilon,b_k+2\varepsilon]$, $b_k \in \mathcal{G}^+(t_2)$, $\hat \theta_{current}^+ (t)$ automatically converges if the moving status never changes to Case 1 for all $t>t_2$, as $\varepsilon$ could be arbitrarily small. Thus, we only need to derive the limiting value in this proof.

There are only two possible limiting values implied by Lemma~\ref{lem:case_1_con}, i.e., $\hat \theta_{goal}^+[j]$ and $\hat \theta_{goal}^+[j+1]$. Without loss of generality, we assume that the limiting value is $\hat \theta_{goal}^+[j]$ given by \eqref{equ:case21}, i.e., $\hat \theta_{goal}^+(t)=\hat \theta_{goal}^+[j]$, $\forall t>t'_2$, where $t'_2$ is a certain value greater than $t_2$. If $\hat \theta_{goal}^+[j]$ stays in $[b_k-2\varepsilon, b_k+2\varepsilon]$, we come to the desired result.

Next, we prove that $\hat \theta_{goal}^+[j]$ falls in $[b_k-2\varepsilon, b_k+2\varepsilon]$ by contradiction. Without loss of generality, we assume $\hat \theta_{goal}^+[j] > b_k+2\varepsilon$, and $\hat \theta_{current}^+(t)$ moves into $[b_k-2\varepsilon,b_k+2\varepsilon]$ from the left. By incorporating (\ref{equ:bar_u_t+1}) into the definition of $\hat \theta_{current}^+ (t)$, we have
{\small\begin{align}
&\hat \theta_{current}^+ (t'_2+1) =\frac{\bar u^+(t'_2+1)}{\bar v^+(t'_2+1)+\delta}=\nonumber\\
&\frac{(1-\alpha(t'_2+1)) \bar u^+(t'_2) + \alpha(t'_2+1) \{y_i \hat h^+(t'_2+1)\}_{avg}}{(1-\alpha(t'_2+1)[\bar v^+(t'_2) + \delta] + \alpha(t'_2+1)[\{ \hat h^+(t'_2+1)\}_{avg} + \delta]}.\nonumber
\end{align}}
Thus, the limiting value of $\hat \theta_{current}^+ (t)$ can be expressed as
\begin{align}
&\lim_{t \rightarrow \infty} \hat \theta_{current}^+ (t) \nonumber\\
& = \frac{\sum_{t = t'_2}^{\infty} \alpha(t)(1-\alpha(t))^{t-t'_2} \{y_i \hat h^+(t)\}_{avg}}{\sum_{t= t'_2}^{\infty} \alpha(t)(1-\alpha(t))^{t-t'_2} [\{\hat h^+(t)\}_{avg} + \delta]}\nonumber\\
&>  \frac{\sum_{t = t'_2}^{\infty} \alpha(t)(1-\alpha(t))^{t-t'_2} [\{\hat h^+(t)\}_{avg} + \delta](b_k +2\varepsilon)}{\sum_{t= t'_2}^{\infty} \alpha(t)(1-\alpha(t))^{t-t'_2} [\{\hat h^+(t)\}_{avg} + \delta]}\nonumber\\
&=b_k +2\varepsilon,
\end{align}
which is in contradiction to the fact that $\hat \theta_{current}^+(t)$ is in $[b_k - 2\varepsilon, b_k + 2\varepsilon]$ for all $t>t_2$. Thus, we conclude that $\hat \theta_{goal}^+[j]$ stays in $[b_k-2\varepsilon, b_k+2\varepsilon]$.
\end{proof}

\end{document}